\documentclass[journal]{IEEEtran}
\IEEEoverridecommandlockouts
\makeatletter
\def\endthebibliography{%
  \def\@noitemerr{\@latex@warning{Empty `thebibliography' environment}}%
  \endlist
}
\makeatother
\usepackage{cite}
\usepackage{algorithm}
\usepackage{setspace}
\usepackage{algpseudocode}
\usepackage{lipsum}
\usepackage{float}
\usepackage{graphicx}
\usepackage{subfig}
\usepackage{balance}
\usepackage{textcomp}
\usepackage{stfloats}
\usepackage{xcolor}
\usepackage{amsmath,amssymb,amsfonts}
\usepackage{float}

\ifodd 1
\newcommand{\com}[1]{{\color{black}#1}} 
\else

\newcommand{\com}[1]{}
\fi

\newtheorem{lemma}{Lemma}
\newtheorem{theorem}{Theorem}
\newtheorem{proof}{Proof}
\newtheorem{remark}{Remark}

\newcommand{\RNum}[1]{\uppercase\expandafter{\romannumeral #1\relax}}

\def\BibTeX{{\rm B\kern-.05em{\sc i\kern-.025em b}\kern-.08em
    T\kern-.1667em\lower.7ex\hbox{E}\kern-.125emX}}
\begin{document}
\font\myfont=cmr12 at 20pt
\title{{\myfont Two-Stage Resource Allocation in Reconfigurable Intelligent Surface Assisted Hybrid Networks via Multi-Player Bandits}\\
\thanks{This work was supported in part by the National Natural Science Foundation of China under Grant 61771017;
in part by the US National Science Foundation under Grants CNS-2107216 and CNS-2128368.
L. Fu is the \emph{corresponding author} (e-mail: liqun@xmu.edu.cn)}
\thanks{J. Tong and L. Fu are with the Department of Information and Communication Engineering, Xiamen University, Xiamen 361005, China (e-mails: tongjingwen@stu.xmu.edu.cn; liqun@xmu.edu.cn).}
\thanks{H. Zhang is with the Department of Electrical and Computer Engineering, Princeton University, NJ, USA (e-mail: hongliang.zhang92@gmail.com).}
\thanks{A. Leshem is with the  Faculty of Engineering, Bar-Ilan University, Ramat Gan 52900, Israel (e-mail: leshema@biu.ac.il).}
\thanks{Z. Han is with the Department of Electrical and Computer Engineering at the University of Houston, Houston, TX 77004 USA, and also with the Department of Computer Science and Engineering, Kyung Hee University, Seoul, South Korea, 446-701, USA (e-mail: zhan2@uh.edu).}
}
\author{\IEEEauthorblockN{Jingwen Tong, ~\IEEEmembership{Student Member,~IEEE,}
Hongliang Zhang, ~\IEEEmembership{Member,~IEEE,}
Liqun Fu, ~\IEEEmembership{Senior Member,~IEEE,}
Amir Leshem, ~\IEEEmembership{Senior Member,~IEEE,} and
Zhu Han, ~\IEEEmembership{Fellow,~IEEE}}}

\maketitle
\begin{abstract}
This paper considers a resource allocation problem where several Internet-of-Things (IoT) devices send data to a base station (BS)
with or without the help of the reconfigurable intelligent surface (RIS) assisted cellular network.
The objective is to maximize the sum rate of all IoT devices by finding the optimal RIS and spreading factor (SF) for each device.
Since these IoT devices lack prior information on the RISs or the channel state information (CSI),
a distributed resource allocation framework with low complexity and learning features is required to achieve this goal.
Therefore, we model this problem as a two-stage multi-player multi-armed bandit (MPMAB) framework to learn the optimal RIS and SF sequentially.
Then, we put forth an exploration and exploitation boosting (E2Boost) algorithm to solve this two-stage MPMAB problem by combining the $\epsilon$-greedy algorithm, Thompson sampling (TS) algorithm, and non-cooperation game method.
We derive an upper regret bound for the proposed algorithm, i.e., $\mathcal{O}(\log^{1+\delta}_2 T)$, increasing logarithmically with the time horizon $T$.
Numerical results show that the E2Boost algorithm has the best performance among the existing methods and exhibits a fast convergence rate.
More importantly, the proposed algorithm is not sensitive to the number of combinations of the RISs and SFs thanks to the two-stage allocation mechanism,
which can benefit high-density networks.
\end{abstract}


\section{Introduction}
Reconfigurable intelligent surface (RIS), which enhances the communication quality by adjusting the amplitude and the phase shift of the incident signal on a 2D planar surface with massive low-cost passive reflecting elements, has drawn increasing attention in future communication networks \cite{wu2019towards, di2020smart, MohamedRISwirelesscommun}.
There have existed some works accounting for this vision by studying the performance of the RIS-assisted cellular network \cite{zhang2020reconfigurable, di2019hybrid}, RIS-assisted unmanned aerial vehicle network \cite{li2020reconfigurable}, and RIS-assisted secure wireless communications \cite{HonglinagBook}.

Meanwhile, the cellular Internet-of-Things (C-IoT) with RIS is regarded as one of the paradigms in future communication networks, providing the capabilities of low-cost, large-scale, and ultra-durable connectivity for everything \cite{liberg2017cellular,qi2018wireless, dama2016feasible}.
By employing the LoRa (short for Long Range) technology,
C-IoT can operate on the unlicensed band since the resulting signal has substantial anti-interference properties \cite{elsaadany2017cellular}.
On the other hand, C-IoT can achieve the rate adaptation by employing the chirp spreading spectrum modulation at the physical layer with different spreading factors (SFs) \cite{waret2018lora}.
However, the study of the network-level performance of these C-IoT devices in the {RIS-assisted} hybrid cellular network still needs more {research}.

In light of this, we consider a hybrid uplink {network} where several C-IoT devices transmit data to a base station (BS)
by opportunistically accessing the RIS-assisted cellular network.
The goal is to maximize the sum rate of all C-IoT devices by finding the best RIS and SF for each device.
Although these C-IoT devices can directly send data to the BS, a higher SF that corresponds to a lower data rate will be assigned to combat the harsh channel environment or to enable a long-range transmission \cite{lyu2019achieving}.
As pointed out in \cite{qi2018wireless}, the low data rate will result in high data latency and security problems.
Therefore, these C-IoT devices may opportunistically access the vacant RISs to improve their data rate by reflecting their signal to the BS.
However, finding the optimal RIS and collecting the exact channel state information (CSI) is challenging for these C-IoT devices.
On the one hand,
the C-IoT device has no information (e.g., the phase shifts) about the RISs since they are deployed for cellular users (UEs).
On the other hand, {there is no communication among C-IoT devices} in such a distributed network.
These features render most traditional optimization methods infeasible in this resource allocation problem, such as the convex optimization methods \cite{boyd2004convex} and the combinatorial optimization methods \cite{papadimitriou1998combinatorial}.

To overcome the above impediments, the learning theory has been considered in \cite{liberg2017cellular, elsaadany2017cellular, nasir2019multi, gu2020deep, sutton2018reinforcement}  to address this problem by sequentially exploring all actions and automatically exploiting the best action.
Refs. \cite{nasir2019multi, gu2020deep} study the distributed resource allocation problem in wireless networks by formulating it as a Markov decision processing (MDP) {problem}.
Then, the authors propose the multi-agent reinforcement learning (RL) based method to solve this MDP problem.
Unfortunately, these solutions often suffer from the issues of the curse of dimensionality, lack of performance guarantee (\textit{e.g.}, the unknown convergence rate), and high computational complexity \cite{sutton2018reinforcement}.
As pointed out in \cite{liberg2017cellular, elsaadany2017cellular}, low complexity and fast convergence resource allocation algorithms are crucial for energy-constrained IoT devices in future communication networks.

This inspires us to consider the multi-armed bandit (MAB) technique.
MAB is a basic framework for the sequential decision-making problem \cite{bubeck2012regret}.
In the classic MAB setting, in each round, a decision-maker (or player) selects an arm from a set of arms (or arm space) with an unknown distribution.
After that, the player will observe a reward from the environment (or the unknown distribution).
The goal is to minimize the pseudo-regret that is defined as the difference between the mean rewards of the optimal arm and the currently selected arm.
During this process, the player faces an exploration and exploitation (EE) dilemma.
On the one hand, the player needs to explore the arm space sufficiently to ensure its long-term performance (i.e., not miss the optimal arm);
on the other hand, it needs to exploit the current best arm as many times as possible to maximize its total rewards.
Compared with the other learning-based methods, MAB has a theoretical guarantee (i.e., regret bound) and low computational complexity, and it is easy to implement.

Recently, the multi-player MAB (MPMAB) framework has gained much attention in wireless communications \cite{ta2019lora01, tibrewal2019distributed, zafaruddin2019distributed, bistritz2018game}.
Ref. \cite{ta2019lora01} studies the SF allocation problem in the LoRa network by devising a fully distributed MPMAB framework.
It solves this MPMAB problem by using the Exponential-weight algorithm for Exploration and Exploitation (Exp3) \cite{Exp3} algorithm.
However, the solution of the Exp3 algorithm is selfish in that it cannot guarantee the optimal allocation for each device.
{The optimal MPMAB framework} is considered in \cite{tibrewal2019distributed},
where different players contend for the same set of channels in an ad-hoc network.
Based on the Hungarian algorithm \cite{jonker1986improving}, the authors propose a probably approximately correct (PAC) based MPMAB algorithm to estimate the CSI matrix sequentially.
However, {the PAC-based MPMAB} algorithm requires players to exchange messages, leading to extra signaling in the system.
The fully distributed resource allocation framework with the optimal solution is investigated in \cite{zafaruddin2019distributed} and \cite{bistritz2018game}.
Ref. \cite{zafaruddin2019distributed} aims to maximize the sum rate of all users by combining the MAB algorithm and the auction algorithm.
A more general version of the distributed MPMAB framework named the game-of-thrones (GoT) algorithm has been proposed in \cite{bistritz2018game}.
The authors intend to find the optimal assignment for each player by combining the {MAB algorithm} and the game theory.
However, algorithms in \cite{zafaruddin2019distributed} and \cite{bistritz2018game} suffer from low convergence rate,
especially when the arm space is large.

In this paper, we propose a two-stage MPMAB framework to {attack} this resource allocation problem in the hybrid uplink network.
In this two-stage MPMAB framework, players are the IoT devices;
{arms are the RISs in the first stage and the SFs in the second stage, respectively.}
We assume that two or more players who select the same RIS will observe a collision and receive zero reward.
This resource allocation problem is quite different from that in \cite{tibrewal2019distributed} and \cite{zafaruddin2019distributed},
because it not only needs to learn the CSI but also the phase shifts of the RISs.
Moreover, the ascending order in the set of SFs and the corresponding descending order in the successful transmission probabilities enable us to devise a two-stage MPMAB framework.
To address this two-stage MPMAB problem, we put forth an exploration and exploitation boosting (E2Boost) algorithm by combining the game theory and the {MAB algorithm}.
The E2Boost algorithm proceeds in epochs and has three phases, i.e., $\epsilon$-greedy EE phase,  non-cooperation game phase, and  Thompson sampling (TS) EE phase.
Each phase contains a specific mechanism to trade off the EE dilemma, which is why we call it the E2Boost algorithm.
{In addition, we derive an upper pseudo-regret bound for the E2Boost algorithm, i.e., $\mathcal{O}(\log^{1+\delta}_2 T)$ where $0\leq\delta<1$, indicating that the per-round regret will trend to $0$ when the time horizon $T$ is sufficiently large.
More importantly, this upper regret bound is about $M$ times lower than that in the GoT algorithm, where $M$ is the number of SFs.
In other words, the proposed algorithm is not sensitive to the number of combinations of the RISs and SFs, which can benefit high-density networks.}

\begin{figure}[!t]
\centering
\includegraphics[width=2.0in]{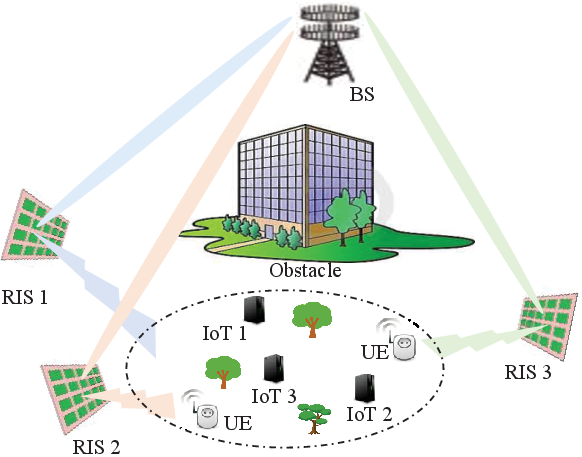}
\caption{A RIS-assisted hybrid uplink network.}
\label{NC}
\end{figure}
{The difference between this work and the existing ones and the main contributions of this work are summarized as follows.}
\begin{itemize}
  \item The E2Boost algorithm embeds the $\epsilon$-greedy algorithm \cite{auer2002finite} in the first phase to reduce the regrets generated from the uniform exploration.
  Specifically, we use the Wasserstein distance (WD) \cite{arras2019bound} to measure the convergence rate of the second phase.
  In return, this measurement is regarded as a criterion to optimize the parameter $\epsilon$.
  \item The E2Boost algorithm adopts the TS algorithm \cite{gupta2018low, agrawal2013further} in the third phase to determine the best SF.
  Since the only observed information is the success or failure transmission feedback, the TS algorithm maintains a Beta distribution on the successful transmission probability of each SF.
  For the Bernoulli reward processing, the TS algorithm accounts for the best performance among the existing stochastic MAB algorithms \cite{agrawal2013further}.
  \item The E2Boost algorithm has a smaller arm space to explore than the GoT algorithm.
  Thanks to the two-stage allocation mechanism, the E2Boost algorithm only needs to explore the sets of the RISs or the SFs. In contrast, the GoT algorithm requires exploring the combinations of the RISs and SFs.
\end{itemize}

The remainder of this paper is organized as follows.
In Section \RNum{2}, we introduce the channel model and the achievable data rate.
The problem formulation is given in Section \RNum{3}.
In Section \RNum{4}, we {introduce} the two-stage MPMAB framework for this joint resource assignment problem.
The E2Boost algorithm is presented in Section \RNum{5}.
Numerical results are given in Section \RNum{6} to evaluate the proposed algorithm.
This paper is concluded in Section \RNum{7}.

\section{System Model}\label{SM}
\setlength{\textfloatsep}{5pt}
We consider a hybrid uplink cellular network, as shown in Fig. \ref{NC},
where several UEs and $N$ IoT devices are located in an area.
Both UE and IoT device need to transmit data to the BS.
Since there may exist some obstacles (\emph{e.g.}, buildings) between UEs and the BS, the signal will experience deep-fading.
Thus, $K$ RISs are deployed to reflect the UEs' signals to the BS by adjusting the RISs' phase shifts.
These RISs are operated over different frequencies\footnote{The RIS can operate at different frequencies by changing the location and the wave-number of each element \cite{di2019hybrid}.}.
The $N$ IoT device has no information about these RISs,
but it may opportunistically access these RISs to improve its data rate.
In this hybrid network, UE is the legal user to communicate with the BS through the RIS;
while the IoT device is required to perform spectrum sensing\footnote{If the received signal strength (RSS) exceeds a threshold, the IoT device marks this RIS with the busy state; otherwise, the state of the RIS is idle.} before access to the vacant RIS.
Time is slotted in $t=1,2,\ldots,T$.
At each time slot, we assume that one RIS can serve multiple UEs but can be exploited by only one IoT device.
{The reason is that the BS requires the precoding and beamforming vectors to maintain communication quality in this multi-user RIS-assisted system \cite{you2021reconfigurable}.
These vectors often contain the information of the CSI and the RISs' phase shifts determined by the UEs.
As a result, an RIS can only support one IoT device since the IoT device lacks these precoding and beamforming vectors.}

\subsection{Channel Models}
There are two transmission patterns for each IoT device {in} this hybrid network.
The first one {is} the RIS-assisted transmission pattern (Pattern \RNum{1}), {where} the IoT device transmits to the BS through the RIS when the target RIS is detected in an idle state.
The second one {is} non-RIS-assisted transmission pattern (Pattern \RNum{2}), {where} the IoT device directly transmits to the BS with a low data rate if the target RIS is detected in a busy state.

\textbf{Pattern \RNum{1}:}
Assume that each element on the RIS is equipped with $b$ PIN diodes, producing $2^b$ phase shifts in $[0, 2\pi)$  by controlling the ON/OFF state of each diode.
Hence, the available phase shift at the $(l_1, l_2)$-th element is
\begin{equation}\label{phase-shift}
\tau_{l_1, l_2} =  \frac{\pi  \rho_{l_1, l_2} }{2^{b-1}},
\end{equation}
where $(l_1, l_2)$ is the index of the RIS elements' matrix and $\rho_{l_1, l_2}$ is an integer in $[0, 2^b-1]$.
Let $A_{l_1, l_2}$ be the reflection factor at the $l_1$-th row and $l_2$-th column of  the RIS elements' matrix, which is defined as
\begin{equation}\label{ReflectionFactor}
A_{l_1,l_2} = A e^{-j\tau_{l_1, l_2}},
\end{equation}
where $A$ is a reflection amplitude with a constant value among $(0, 1]$\footnote{The reflection amplitude can be a function of the phase shift as in \cite{abeywickrama2020intelligent}.}.

By taking advantage of the directional reflections of the RIS,
the BS - RIS - IoT device link is usually stronger than other multi-paths as well as the deep-fading direct link between the BS and the IoT device \cite{HonglinagBook}.
Therefore, we model the channel between the BS and the IoT device as a Ricean model.
In this way, the {BS - (RIS $k$) - (IoT device $n$)} link acts as the dominant ``LoS'' component;
while all the other paths together form the ``non-LoS (NLoS)'' component.
Hence, the RIS-assisted channel model $h^{n,k}_{l_1,l_2}$ is defined as
\begin{equation}\label{ChannelModel}
h^{n,k}_{l_1,l_2} =\sqrt{\frac{\zeta}{\zeta+1}} \tilde{h}^{n,k}_{l_1,l_2} + \sqrt{\frac{1}{\zeta+1}} \hat{h}^{n,k}_{l_1,l_2},
\end{equation}
where $\tilde{h}^{n,k}_{l_1,l_2}$ and $\hat{h}^{n,k}_{l_1,l_2}$ are the  LoS component and the NLoS component with the $k$-th RIS and the $n$-th IoT device through the $(l_1,l_2)$-th element, respectively.
Symbol $\zeta$ is the Rician factor, indicating the ratio of the LoS component to the NLoS component.
In the following, we omit the IoT device index $n$ and the RIS index $k$ in the superscript if no confusion occurs.

Let $D_{l_1,l_2}$ be the distance between the BS and the $(l_1, l_2)$-th RIS element,
and let $d_{l_1,l_2}$ be the distance between the {$(l_1, l_2)$-th RIS element} and the IoT device.
The transmission distance of {BS - ($(l_1, l_2)$-th RIS element) - (IoT device $n$)} link is $L_{l_1,l_2} = D_{l_1,l_2}+d_{l_1,l_2}$.
According to \cite{zhang2020reconfigurable}, {the LoS component of this link} is given by
\begin{equation}\label{ChannelModel01}\small
\begin{split}
\tilde{h}_{l_1,l_2}
&= \sqrt{G D^{-\iota}_{l_1,l_2} d^{-\iota}_{l_1,l_2}} e^{-j \frac{2\pi}{\lambda} L_{l_1, l_2}}\\
&= \sqrt{G}\left[\sqrt{D^{-\iota}_{l_1,l_2}} e^{-j \frac{2\pi}{\lambda} D_{l_1, l_2}} \right] \left[ \sqrt{d^{-\iota}_{l_1,l_2}} e^{-j \frac{2\pi}{\lambda} d_{l_1, l_2}} \right],
\end{split}
\end{equation}
where $\iota$ is the path-loss exponent.
Symbol $G$ is the antenna gain, and $\lambda$ is the wavelength of the signal.
Meanwhile, the NLoS component is given by
\begin{equation}\label{ChannelModel02}
\hat{h}_{l_1,l_2} = \sqrt{PL_{\mathrm{NLoS}}(L_{l_1,l_2})}g_{l_1,l_2},
\end{equation}
where $g_{l_1,l_2}$ is the small-scale NLoS component, following the i.i.d. complex Gaussian distribution, i.e., $g_{l_1,l_2} \sim \mathcal{CN}(0,1)$.
Term $PL_{\mathrm{NLoS}}(\cdot)$ is the NLoS channel power gain that we adopt the urban macro (UMa) path-loss model\footnote{The calculation of $PL_{\mathrm{NLoS}}(\cdot)$ in dB form is $10 \log_{10} PL_{\mathrm{NLoS}}(d)  = 13.54 + 39.08\log_{10} (d) + 20\log_{10}(f_c) - 0.6(h_{\mathrm{IoT}}-1.5)$, where $d$ is the Euclidean distance between the device and the BS, and $h$ is the height of the device. Symbol $f_c$ is the central frequency.}  \cite{3GPP} in the simulation.

\textbf{Pattern \RNum{2}:}
In the non-RIS-assisted transmission pattern, the IoT device directly transmits to the BS without the help of the RIS.
Since there are some obstacles, the signal may {experience} deep fading.
{Thus, we use the shadow fading model to describe the channel between IoT device $n$ and the BS \cite{tong2018cooperative}, i.e.,}
\begin{equation}\label{ChannelModel03}
h_n = \sqrt{\varrho_n} g_n,
\end{equation}
where $\varrho_n$ is the channel power gain, following the i.i.d. log-normal distribution with mean $\mu_n$ and standard deviation $\sigma_n$ of $\ln \left(\varrho_n\right)$.
The typical value of $\sigma_n$ is between $6$ and 12 dB for practical radio channels \cite{3GPP}.
In addition, $g_n$ is a small-scale NLoS component, following i.i.d. complex Gaussian distribution, i.e., $g_n \sim \mathcal{CN}(0,1)$.

\subsection{Signal Model and Achievable Data Rate}
The received signal from IoT device $n$ to the BS through RIS $k$ is given by
\begin{equation}\label{SignalModel}
\left\{\begin{array}{lll}
\digamma_{n,k} = \sum_{l_1,l_2} A^k_{l_1,l_2} h^{n,k}_{l_1,l_2} \sqrt{\Omega_{n}} x + y + \omega,  & \mathrm{Pattern \ \RNum{1}}, \\
\digamma_{n} =  h_{n} \sqrt{\Omega_{n}} x + y + \omega, & \mathrm{Pattern \ \RNum{2}},
\end{array}\right.
\end{equation}
where $x$ is the transmission signal with $|x|^2=1$ and $y$ is the received interference signal\footnote{Notice that the interference $y$ may come from the neighboring cellular networks or the local UEs when the UEs' signals are missed detection by the IoT devices.} which can be modeled as a log-normal distribution \cite{tong2017energy}, i.e., $y \sim Log\mathcal{N}(\mu_y,\sigma^2_y)$.
In addition, $\omega \sim \mathcal{CN}(0, \sigma_{\omega}^2)$ is the i.i.d. additive complex Gaussian noise and $\Omega_n$ is the transmit power.
Then, the received signal-to-interference-plus-noise ratio (SINR) {can be calculated by}
\begin{equation}\label{SNR}\small
\left\{\begin{array}{lll}
\gamma_{n,k} &= \frac{\Omega_n \left(\sum_{l_1,l_2}A^k_{l_1,l_2} \tilde{h}^{n,k}_{l_1,l_2} \sum_{l_1,l_2} \left(A^{k}_{l_1,l_2}\right)^{\ast} \left(\tilde{h}^{n,k}_{l_1,l_2}\right)^{\ast} \right) }{\exp\left(2\mu_y + 2\sigma^2_y \right) + \sigma_{\omega}^2}, &  \\
\gamma_{n} &= \frac{\Omega_n h_n h^{\ast}_n }{\exp\left(2\mu_y + 2\sigma^2_y \right) + \sigma_{\omega}^2}, &
\end{array}\right.
\end{equation}
where $(\cdot)^{\ast}$ is the conjugate operation.

In a practical system, each IoT device can only support a finite number of data rates according to the available SFs.
Let $\mathcal{M} = \{c_1, c_2, \cdots, c_M\}$ and $\mathcal{S}=\{s_1, s_2, \cdots, s_M \}$ be the set of data rates and SFs, respectively.
According to \cite{waret2018lora}, the relationship between data rate and SF is given by
\begin{equation}\label{SF}
c_m = \frac{B s_m }{2^{s_m}}\times CR,
\end{equation}
where $B$ is the bandwidth in Hz and $CR$ is the code rate.
It can be seen that a higher SF is associated with a lower data rate.
In other words, if $\mathcal{S}$ is in ascending order $s_1<s_2<\cdots<s_M$, $\mathcal{M}$ will be the descending order $c_1>c_2>\cdots>c_M$.

The achievable data rate not only corresponds to the selected modulation and coding scheme but also depends on the received SINR \cite{VL1}.
Thus, according to the instantaneous received SINR,
the successful transmission probability of data rate $c_m$ is given by
\begin{equation}\label{Succ}
\left\{\begin{array}{lll}
\theta_{k, c_m}^{n} &\triangleq \text{Pr}\{\gamma'_{n,k} \geq \Psi_m \}, & \mathrm{Pattern \ \RNum{1}}, \\
\theta_{c_m}^{n} &\triangleq \text{Pr}\{\gamma'_{n} \geq \Psi_m \}, & \mathrm{Pattern \ \RNum{2}},
\end{array}\right.
\end{equation}
where $\Psi_m$ is the minimum required SINR for the BS to demodulate the received signal {when the data rate is $c_m$}.
Note that SINR $\gamma'_{n,k}$ (or $\gamma'_{n}$) is a random variable with mean $\gamma_{n,k}$ (or $\gamma_{n}$) due to the small-scale NLoS components and the received interference signal.
According to the Shannon formula, {the better received SINR, the higher the successful transmission probability when given a data rate.
In other words, a descending data rates ($c_1>c_2>\cdots>c_M$) will lead to an ascending successful transmission probabilities ($\theta_{c_1}<\theta_{c_2}<\cdots<\theta_{c_M}$).}

\section{Problem Formulation}\label{PF}
The system's goal is to maximize the sum rate of all IoT devices at each time slot by finding the optimal RIS and SF for each device
under Pattern \RNum{1}, as well as determining the optimal SF for each device under Pattern \RNum{2}.
Let $\vec{\vartheta}^t=\{\vartheta^t_1, \vartheta^t_2, \ldots, \vartheta^t_K \}$ be the state vector of the RISs at time slot $t$,
{where} $\vartheta^t_k=1$ means that the $k$-th RIS is vacant; otherwise, it is occupied.
{Note that this information is known prior to each IoT device with the spectrum sensing operation.}
Then, the resource allocation problem is given by
\begin{equation}\label{DSO}\small
\begin{aligned}
& \underset{}{\max\limits_{\phi^n_{k, c_m}, \psi^n_{c_m}}}
& & \sum_{t=1}^{T}\sum_{n=1}^{N} \sum_{m=1}^{M} \left( \underbrace{\sum_{k=1}^{K} c_m \vartheta^t_{k}  \theta^n_{k, c_m} \phi^n_{k, c_m}}\limits_{\mathrm{Pattern \ \RNum{1}}}
+ \underbrace{ c_m  \theta^n_{c_m} \psi^n_{c_m}}\limits_{\mathrm{Pattern \ \RNum{2}}}\right) \\
& \mathrm{s.t.}
& & \sum_{m=1}^{M} \sum_{k=1}^{K} \vartheta^t_{k} \phi^n_{k, c_m}  +  \sum_{m=1}^{M} \psi^n_{c_m} =1,  \ \forall n \in \mathcal{N},  \\
& \mathrm{\quad}
& & \sum_{n=1}^{N} \phi^n_{k, c_m} \leq 1, \  \forall c_m \in \mathcal{M}, \ \text{and} \ \forall k \in \mathcal{K},
\end{aligned}
\end{equation}
where $\phi^n_{k, c_m}$ and $\psi^n_{c_m}$ are the binary variables,
where $\phi^n_{k, c_m} =1$ denotes that IoT device $n$ transmits on the $k$-th RIS with SF $s_m$ to reflect its signal to the BS; otherwise, $\phi^n_{k, c_m} =0$.
The symbol $\psi^n_{c_m} =1$ denotes that IoT device $n$ directly transmits to the BS with SF $s_m$; otherwise, $\psi^n_{c_m} =0$.
Thus, the first constraint indicates that each IoT device either transmits on Pattern \RNum{1} or Pattern \RNum{2}.
If IoT device $n$ transmits on Pattern \RNum{1}, then  $\sum_{m=1}^{M} \sum_{k=1}^{K} \vartheta^t_{k} \phi^n_{k, c_m} =1$ means that each IoT device can only select a {pair of RIS and SF};
if IoT device $n$ transmits on Pattern \RNum{2}, then $\sum_{m=1}^{M} \psi^n_{c_m} =1$ denotes that each IoT device {can only} select a SF.
The second constraint means that the number of IoT devices that select the $k$-RIS and the $m$-th SF is subject to $0$ or $1$.
In addition, $\mathcal{N} = \{1, 2, \cdots, N\}$ and $\mathcal{K} = \{1,2,\cdots, K \}$ are the sets of IoT devices and RISs, respectively.
The symbol $\theta^n_{k, c_m}$ is the successful transmission probability that IoT device $n$ transmits on the $k$-th RIS and the $m$-th SF;
while $\theta^n_{c_m}$ is the successful transmission probability that IoT device $n$ directly sends data to the BS with SF $m$.

It is difficult to solve problem \eqref{DSO} in this distributed hybrid network, especially in Pattern \RNum{1}.
First, $c_m$ and $\theta^n_{k, c_m}$ are discrete values\footnote{According to \eqref{SF}, $c_m$ is discrete since the number of SFs is limited in practice. In addition, according to \eqref{SNR} and \eqref{Succ}, $\theta^n_{k, c_m}$ is a function of the RIS's phase shifts, which are discrete values in range $[0, 2\pi)$.},  resulting in a non-convex problem.
Second, it requires the exact value of $\theta^n_{k, c_m}$.
This information is difficult to obtain since the channel characteristic is determined by the UE-controlled RISs.
Third, it needs some communications among IoT devices to share the information of $\theta^n_{k, c_m}$ {so as to} determine the optimal available RIS for each IoT device.
In addition, {problem \eqref{DSO} in Pattern \RNum{2}} can be regarded as a rate adaptation problem \cite{Gupta2018Low-complexity} since the SF allocation is independent for each IoT device,
{but} it still requires the exact CSI (or the value of $\theta^n_{c_m}$), which is hard to estimate from the time-varying channel.

To overcome these challenges, we adopt the online learning method to learn the values of $\theta^n_{c_m}$ and $\theta^n_{k, c_m}$ sequentially and to allocate the optimal RIS and SF to each IoT device adaptively.
During this process, the IoT device not only needs to explore the combinations of the RISs and SFs sufficiently
but also needs to exploit the current best RIS and SF as many times as possible at each time slot.
To better tradeoff this EE dilemma, we introduce the MPMAB framework to solve this problem,
{where players are the IoT devices, and arms are the combinations of the RISs and SFs.}

However, the MPMAB framework still suffers from a slow convergence rate due to the large arm space (i.e., the combinations of the RISs and SFs) under Pattern \RNum{1}.
Therefore, we decouple the MPMAB problem into a two-stage MPMAB framework to shrink the feasible arm space.
{The reason is that, on the one hand,} descending data rates will result in ascending successful transmission probabilities;
{on the other hand,} an IoT device with different data rates will experience the same channel fading under a particular RIS.
These features indicate that the average successful transmission probabilities of the ordered data rates over different RISs have the same trend.
Therefore, we can explore these RISs by arbitrarily assigning a data rate to the IoT device.
In other words, the SF allocation and the RIS allocation processes are independent of each other under Pattern \RNum{1}.

\section{Two-Stage MPMAB-based Resource Allocation Framework}\label{TSMF}
In this two-stage MPMAB framework, players are the IoT devices;
arms are the RISs and the SFs in the first and second stages, respectively.
The first-stage MPMAB problem is to determine the best RIS for each IoT device;
while the second-stage MPMAB problem is to find the optimal SF based on the state of the {determined} RIS.

\subsection{First-Stage MPMAB Framework}
{We first introduce the transmission feedback model and the collision model.}
In the transmission feedback model, the IoT device can receive the transmission feedback from the BS {when} it transmits on Pattern \RNum{1} or Pattern \RNum{2}.
Specifically, let $I'_{n,t}$ be the selected arm by the $n$-th IoT device at time slot $t$.
After transmitting on the $I'_{n,t}$-th \com{arm}, the IoT device $n$ will receive transmission feedback $X_{I'_{n,t}}(t)$ from the BS.
If the transmission is successful, \com{then} $X_{I'_{n,t}}(t)=1$; otherwise, $X_{I'_{n,t}}(t)=0$.

{The collision model only exists in the first stage,} referring to that two or more IoT devices that choose the same RIS will receive no rewards.
We assume that each IoT device can deduce this collision information by observing the timeout feedback flag.
Specifically, let $\eta$ be the collision indicator.
If an IoT device does not receive any feedback from the BS in the current time slot, then a collision happens, i.e., $\eta=0$; otherwise, $\eta=1$.
Therefore, IoT devices can distinguish the collision and the transmission failure events by checking whether or not they receive transmission feedback from the BS.
Moreover, this collision model also works in some extreme situations.
{For example, when} the received SINR in Pattern \RNum{1} is too low to be recognized by the BS,
the IoT device can always set $\eta=0$ (i.e., the reward is $0$) since the target RIS is suboptimal to it.

Denote $\boldsymbol{I}'_t = \{I'_{1,t},I^{'}_{2,t}\cdots, I'_{N,t}\}$ by the strategy profile at time slot $t$.
The collision indicator of RIS $k$ is defined as
\begin{equation}\label{Indicator}
\renewcommand{\arraystretch}{1.3}
\begin{split}
\eta_k \left( \boldsymbol{I}'_t \right) = \left\{
\begin{array}{ll}
0, & |\mathcal{N}_k| > 1, \\
1, & \text{otherwise},
\end{array} \right.
\end{split}
\end{equation}
where $\mathcal{N}_k$ is the set of players that select the $k$-th RIS in strategy profile $\boldsymbol{I}'_t$.
The reward that IoT device $n$ transmits on the $k$-th RIS is given by
\begin{equation}\label{Reward01}
r_{n, I'_{n,t}=k} (t) \triangleq  \eta_k \left(\boldsymbol{I}'_t\right) X_{I'_{n,t}=k}(t).
\end{equation}
Then, the estimated average successful transmission probability that the $n$-th IoT device transmits on the $k$-th RIS is given by
\begin{equation}\label{Ave_Reward01}
\begin{split}
\hat{\theta}_{n,k} &= \mathbb{E}\left[r_{n, I'_{n,t}=k} (t)\right],\\
\end{split}
\end{equation}
where $\mathbb{E}[\cdot]$ is the expectation operator.

\subsection{Second-Stage MPMAB Framework}
In this stage, each player has a targeted RIS {after} the first-stage allocation.
Thus, the IoT device transmits directly or on a targeted RIS to the BS at each time slot.
Note that there is no collision in this stage since two or more players can choose the same SF.
{Therefore, this stage can be regarded as a single-player MAB framework.}

Let $I''_{n,t}$ be the currently selected arm at the second-stage allocation.
The reward that the IoT device $n$ chooses the $m$-th data rate is defined as
\begin{equation}\label{Reward02}
r_{n, I^{''}_{n,t}=m} (t) \triangleq c_{m}  \eta_k \left(\boldsymbol{I}'_t\right) X_{I''_{n,t}=m}(t),
\end{equation}
where $c_m$ is the $m$-th data rate and $\boldsymbol{I}'_t$ is the strategy profile of all players' {target} RISs at the first stage.
The instantaneous rewards $r_{n, I''_{n,t}} (t)$ are independently and identically distributed w.r.t. player $n$ and time slot $t$.
Therefore, the estimated average reward is given by
\begin{equation}\label{Ave_Reward}
\hat{\mu}_{n,m} = \mathbb{E}\left[r_{n, I''_{n,t}=m} (t)\right] = c_m \hat{\theta}_{n,m},
\end{equation}
where $\hat{\theta}_{n,m}$ is the estimated average successful transmission probability that the $n$-th IoT device transmits on the $m$-th data rate.

\subsection{Performance Metric for the Two-Stage MPMAB Framework}
In the following, we design a criterion to quantify the performance loss that players select the suboptimal arms rather than the optimal arm in this two-stage MPMAB problem.
According to \eqref{DSO}, the objective function consists of Pattern \RNum{1} and Pattern \RNum{2}.
For Pattern \RNum{1}, we define the joint RIS and SF selection profile by $\boldsymbol{a} = \{a_{1}, a_{2}, \ldots, a_{N}\}$, {where $a_n \in \mathcal{K} \otimes \mathcal{M}$ and $\otimes$ is the Cartesian product of the RIS set and the data rate set.}
However, for Pattern \RNum{2}, the selection profile $\boldsymbol{a}$ is the set of data rates, i.e., $a_n\in \mathcal{M}$.
Therefore, the two-stage allocation aims to solve the following problem,
\begin{equation}\label{obj0}
\begin{split}
\boldsymbol{a}^{\ast} &= \text{arg} \max\limits_{\boldsymbol{a}} \sum\limits^N_{n=1} \hat{\mu}_{n,a_{n}}\\
&= \text{arg} \max\limits_{\boldsymbol{a}} \sum\limits^N_{n=1} \mathbb{E}\left[ c_{\boldsymbol{a}} \eta_k \left(\boldsymbol{a}\right) X_{\boldsymbol{a}}(t) \right],
\end{split}
\end{equation}
where $\boldsymbol{a}^{\ast} = \{a_1^\ast, a_2^\ast, \cdots, a_N^\ast\}$ is the optimal strategy profile.

Then, we define the difference between the optimal arm and the currently selected arm as the performance metric, also known as \emph{regret}.
According to \cite{bistritz2018game}, the expression of accumulated regrets is given by
\begin{equation}\label{Regret}
\mathcal{R}eg \triangleq  \sum_{t=1}^T \sum_{n=1}^{N} {r}_{n, a_n^{\ast}}(t) - \sum_{t=1}^T \sum_{n=1}^{N} {r}_{n,a_n} (t),
\end{equation}
where $a^\ast_n \in \boldsymbol{a}^\ast$ and $T$ is the total time slots.
For mathematical analysis, we {further} define the \emph{pseudo-regret} \cite{bubeck2012regret} \emph{w.r.t.} the stochastic rewards and the randomly selected arms as
\begin{equation}\label{ERegret}\small
\renewcommand{\arraystretch}{1.6}
\begin{split}
\overline{\mathcal{R}eg} &= \sum_{n=1}^{N} \left( T \times \mu_{n,a^\ast_n} -  \mathbb{E} \sum_{t=1}^T \mu_{n, a_n}  \right)\\
&= \left\{
\begin{array}{ll}
    \sum_{n=1}^{N} \sum_{i=1}^{K\times M} \Delta_{n,i} \mathbb{E} [W_{n,i}], & \mathrm{Pattern} \ \mathrm{\RNum{1}}, \\
    \sum_{n=1}^{N} \sum_{i=1}^{M} \Delta_{n,i} \mathbb{E} [W_{n,i}], & \mathrm{Pattern} \ \mathrm{\RNum{2}},
    \end{array} \right.
\end{split}
\end{equation}where $\Delta_{n,i} = \mu_{n,a^\ast_n} - \mu_{n,i}$ and $W_{n,i}$ is the number of times that arm $i$ has been selected up to time $T$.
Term $\mu_{n, i}$ is the real expected throughput of player $n$ at arm $i$.

\section{E2boost Algorithm}\label{PA}
In this section, we propose an E2Boost algorithm to solve this two-stage MPMAB problem by combining the game theory and the {MAB} algorithms.
The structure of the E2Boost algorithm is shown in Fig. \ref{AS}.
Since time horizon $T$ is unknown to each player,
the E2Boost algorithm proceeds in epochs (i.e., $z = 1, \cdots, Z$).
Each epoch consists of three phases: $\epsilon$-Greedy EE, non-cooperation game, and Thompson sampling EE phases.
Each phase contains several time slots and specific mechanisms to balance the EE dilemma.
\begin{figure}[!t]
\centering
\includegraphics[width=3.3in]{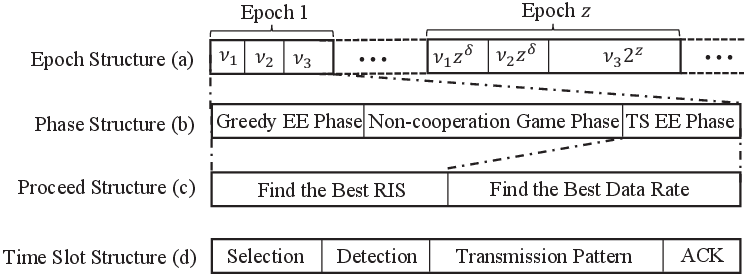}
\caption{The structure of the {E2Boost algorithm}.}
\label{AS}
\end{figure}

\subsection{The Exploration and Exploitation Boosting Algorithm}
The E2Boost algorithm is shown in Algorithm \ref{GoT}.
The first two phases are designed to find the optimal RIS for each IoT device by solving the first-stage MPMAB problem;
while the last phase is to determine the best SF by solving the second-stage MPMAB problem.
In the following, we elaborate on the above three phases in detail.
\begin{algorithm*}[t]
\caption{E2Boost Algorithm Run by Player $n$}
\label{GoT}
\begin{algorithmic}[1]
\State \textbf{Initialization:}  $\delta>0,  \varepsilon >0$ and $\nu, \nu_1, \nu_2, \nu_3 > 0$. Let $\epsilon =1, V_{n,k}(0) = 0, Q_{n,k}(0) = 0, \alpha_{n,m}(0) =0, \beta_{n,m}(0) = 0, \ \forall k \in \mathcal{K}$, \ $\forall m \in \mathcal{M}$.
\For {each epoch $z=1,2,\cdots$, Z}
    \State \textbf{i)  $\epsilon$-Greedy EE Phase:} \ For the next $\nu_1 z^{\delta}$ time slots.
        \State \quad a) Pick up a data rate $I'_{n,t}=m$ uniformly from set $\mathcal{M}$ if $z = 1$, otherwise $I'_{n,t} = c_n^{\ast}$;
        \State \quad b) Select a RIS $I'_{n,t}=k$ uniformly from set $\mathcal{K}$ with probability $\epsilon$ or $I'_{n,t} = k_n^\ast$ with probability $1 - \epsilon$;
        \State \quad c) Detect the selected RIS: jump to \textbf{Phase iii} if busy; otherwise, continue the following steps:
        \State \quad d) Observe the transmission feedback $X_{n, I'_{n, t}}$. Set $\eta_k = 0$ if timeout and $\eta_k = 1$ otherwise;
        \State \quad e) If $\eta_k = 1$ then update $V_{n, I'_{n, t}}(t) = V_{n, I'_{n, t}}(t-1) + 1$ and $Q_{n, I'_{n, t}} (t) = Q_{n, I'_{n, t}} (t-1) + X_{n, I'_{n, t}}(t) $;
        \State \quad At the end of this phase, compute the successful transmission probabilities of RISs by
        \begin{equation*}
        \hat{\theta}^z_{n, k} = \frac{Q_{n, k}}{V_{n, k}}, \quad \forall k \in \mathcal{K}.
        \end{equation*}
    \State \textbf{ii)  Non-cooperation Game Phase:} \ For the next $\nu_2 z^{\delta}$ time slots, play with the dynamics. Set $ST_n = C$, and let $\bar{k}$ be the last RIS chosen in the $z-\lfloor \frac{z}{2} \rfloor-1$ Game phase, or a random choice if $z = 1, 2$.
        \State \quad a) If $ST_n = C$ choose a RIS $I'_{n, t}$ using \eqref{Pr_content} and if $ST_n = D$ select $I'_{n, t}$ at random \eqref{Pr_discontent};
        \State \quad b) Detect the selected RIS: jump to \textbf{Phase iii} if busy;  otherwise continue the following steps:
        \State \quad c) If $I'_{n,t} \neq \bar{k}$ or $u_n = 0$ or $ST_n = D$ then set $ST_n = C \ \text{or} \ D$ according to \eqref{State_D};
        \State \quad d) Record the number of times each RIS \com{has} been selected within the content state using \eqref{Num};
        \State \quad e) Adjust parameter $\epsilon$ according to \emph{Lemma} \ref{Lemma01} when $z\geq 2$.
        \State \quad At the end of this phase, determine the current best RIS by
        \begin{equation*}
        k_n^{\ast} = \text{arg}  \max\limits_{k \in \mathcal{K}} \ \sum_{j = 0}^{\lfloor\frac{z}{2}\rfloor} F^{z-j}_n(k).
        \end{equation*}
    \State \textbf{iii) Thompson Sampling EE Phase:} For the next $\nu_3 2^z$ time slots, run the Thompson sampling algorithm based on the current state of the best RIS, as well as the corresponding collision indicator.
        \State \quad a) Draw $\hat{\theta}_{n,m} \sim \text{Beta}\left(\alpha_{n,m}(t)+1, \beta_{n,m}(t)+1\right)$;
        \State \quad b) Select a data rate $I''_{n,t} = \mathrm{arg}\max_{m \in \mathcal{M}} c_m \times \hat{\theta}_{n,m}$;
        \State \quad c) Detect the target RIS: The device directly transmits to the BS if busy, otherwise continue the following steps:
        \State \quad d) Transmit on the selected data rate and observe the random transmission feedback $X_{I''_{n,t}}(t)$;
        \State \quad e) Posterior update: Set $\alpha_{n,I''_{n, t}}(t) = \alpha_{n, I''_{n, t}}(t-1) + X_{I''_{n, t}}(t)$ and $\beta_{n,I''_{n, t}}(t) = \beta_{n,I''_{n, t}}(t-1) + 1-X_{I''_{n, t}}(t)$.
        \State \quad At the end of this phase, determine the current best data rate by
        \begin{equation*}
        c_n^{\ast} = \text{arg}  \max\limits_{m \in \mathcal{M}} \  \frac{c_m \alpha_{n,m}}{\alpha_{n,m}+\beta_{n,m}}.
        \end{equation*}
\EndFor
\end{algorithmic}
\end{algorithm*}

\textbf{$\epsilon$-Greedy EE Phase}:
There are $ \nu_1 z^\delta$ rounds in this phase for epoch $z = 1, \cdots, Z$, where $\nu_1>0$ and $\delta>0$ are two constants.
It aims to estimate the average successful transmission probability of each RIS.
The SF is randomly chosen from the set $\mathcal{S}$ when $z = 1$; otherwise, it uses the SF determined in the last epoch of the third phase.
We \com{adopt} the $\epsilon$-greedy algorithm to balance the EE dilemma.
Specifically, if $z = 1$, we set $\epsilon = 1$ to uniformly explore all RISs;
otherwise, we update the parameter $\epsilon$ according to \emph{Lemma} \ref{Lemma01}, as given in the next paragraph.
Hence, when the players' strategy profile deviates from $\boldsymbol{a}^\ast$, Algorithm \ref{GoT} tends to uniformly explore all actions;
otherwise, it inherits the last epoch's action with a high probability.

\textbf{Non-cooperation Game Phase}:
This phase has a length of $ \nu_2 z^\delta$ rounds, which is the core step of Algorithm \ref{GoT} to allocate the optimal RIS for each player.
By adopting the estimated average successful transmission probability $\hat{\theta}^z_{n, k}$ in the first phase as a utility, players in this phase will play a non-cooperation game.

Specifically, let the utility of player $n$ in strategy profile $\boldsymbol{I}'$ be
\begin{equation}\label{Utility}
u_n (\boldsymbol{I}') \triangleq \eta_{k}(\boldsymbol{I}') \hat{\theta}^z_{n, k}, \ \forall k \in \mathcal{K},
\end{equation}
where $\hat{\theta}^z_{n,k}$ is the estimated successful transmission probability that the $n$-th IoT device transmits on the $k$-th RIS from epoch $1$ to $z$ at the first phase.
Let $u_{n,\max} = \max\limits_{\boldsymbol{I}'} u_n(\boldsymbol{I}')$ be the maximum utility of player $n$.
\com{Assume that each player} has a private state $ST_n = \{C, D\}, \ \forall n \in \mathcal{N}$, where $C$ and $D$ represent content and discontent state, respectively.
In addition, each player maintains a baseline RIS $\bar{k}$.
\com{Then, a player chooses a RIS according to the following strategy:}
\begin{itemize}
  \item A content player has a very high probability of staying at the current baseline RIS:
    \begin{equation}\label{Pr_content}
    \renewcommand{\arraystretch}{1.2}
    \begin{split}
    P_{n,k} = \left\{
    \begin{array}{ll}
    \frac{\varepsilon^\nu}{K - 1}, & k \neq \bar{k}; \\
    1 - \varepsilon^\nu, & k = \bar{k}.
    \end{array} \right.
    \end{split}
    \end{equation}

  \item A discontent player selects a RIS following a uniform distribution, i.e.,
  \begin{equation}\label{Pr_discontent}
  P_{n,k} = \frac{1}{K}, \ \forall k \in \mathcal{K}.
  \end{equation}
  \end{itemize}
  The transition between content state $C$ and discontent state $D$ is given by:
\begin{itemize}
  \item If $k = \bar{k}$, $u_n>0$, and $ST_n = C$, then a content player keeps state $C$ with \com{a probability of 1}:
  \begin{equation}\label{State_C}
  (\bar{k}, C) \rightarrow (\bar{k}, C).
  \end{equation}

  \item If $k \neq \bar{k}$ or $u_n = 0$ or $ST_n = D$, then transitions of baselines and states are given by
    \begin{equation}\label{State_D}
    \renewcommand{\arraystretch}{1.2}
    \begin{split}
    \left(\bar{k}, C/D\right) = \left\{
    \begin{array}{ll}
    \left(k, C\right), & \frac{u_n}{u_{n,\max}} \varepsilon^{u_{n,\max} - u_n}; \\
    \left(k, D\right), & 1-\frac{u_n}{u_{n,\max}} \varepsilon^{u_{n,\max} - u_n}.
    \end{array} \right.
    \end{split}
    \end{equation}
\end{itemize}
Eq. \eqref{State_D} indicates that, when a RIS \com{records a collision} or in a busy state (i.e., $\eta_k=0$),
the player will transfer to the discontent state $D$ with \com{a probability of $1$ as} $u_n=0$.
On the other hand, when a RIS is optimal to the player, it will transfer to the content state $C$ with \com{a probability of $1$ as} $u_n= u_{n,\max}$.

Assume that all players' actions and states constitute a strategy profile $\boldsymbol{a_1}$.
\com{Then, a strategy graph can be constructed at the end of this phase,}
where the vertex is the strategy profile, and an edge exists if the players can switch from one strategy to the other.
Actually, this strategy graph forms a perturbed time-reversible Markov process over state space $\prod_{n=1}^{N}\left(\mathcal{K}_n \times (C, D)\right)$.
As pointed out in \cite{menon2013convergence, bistritz2018game, marden2014achieving}, there exists an optimal strategy profile that players will visit many times than other strategy profiles.
As a result, each player can agree on \com{its} optimal arm by recording the number of times that each arm has been selected, i.e.,
\begin{equation}\label{Num}
F^z_n (k) \triangleq \sum_{t\in \mathcal{G}_z} \mathbb{I} \left( I'_{n,t} = k, ST_n = C \right),  \ \forall k \in \mathcal{K},
\end{equation}
where $F^z_n (k)$ is the number of times that the $k$-th RIS has been played by the $n$-th player at the $z$-th epoch \com{under state $C$}.
The symbol $\mathcal{G}_z$ represents \com{the number of time slots} in the $z$-th epoch and $\mathbb{I}(\cdot)$ is an indicator function.
Finally, we can determine the best RIS by using the recent $\lfloor z/2 \rfloor +1$ epochs' $F^z_n$, i.e.,
\begin{equation}\label{Num01}
k_n^{\ast} = \text{arg}  \max\limits_{k \in \mathcal{K}} \ \sum_{j = 0}^{\lfloor\frac{z}{2}\rfloor} F^{z-j}_n(k).
\end{equation}

\com{In addition, $F^z_n$ can be used to design a criterion to balance the EE dilemma in the first phase by adjusting the parameter $\epsilon$.}
The reason is that it is unnecessary to uniformly explore all RISs when the non-cooperation game phase asymptotically approaches the optimal RIS.
This asymptotical behavior can be quantified by the distance between two adjacent vectors of $F^z_n$ and $F^{z-1}_n$,
which can be measured by the WD\footnote{Wasserstein distance, also known as earth mover's distance, is a measure to calculate the distance between two probability distributions on a metric space. In the simulation, we compute it using the corresponding function in Matlab.}.
As a result, the distance between the probability mass functions (PMFs)\footnote{The PMF is computed by $\mathbb{P}(F^z_n(i))= F^z_n(i)/ \sum_i F^z_n(i),  \ \forall i \in \mathcal{K}$.} of $F^z_n$ and $F^{z-1}_n$ is regarded as a criterion to adjust the parameter $\epsilon$.
Thus, we have the following lemma.
\begin{lemma}\label{Lemma01}
For the $n$-th player, given $z>1$,  the parameter of the $\epsilon$-greedy algorithm in the first phase can be chosen according to
\begin{equation*}
\epsilon \triangleq  \min \{1, \text{D}_{\small{\text{WD}}} \left( \mathbb{P}(F^z_n) \ || \ \mathbb{P}(F^{z-1}_n) \right) \},
\end{equation*}
where $\text{D}_{\text{WD}} \left( \cdot \ || \ \cdot \right)$ is the calculation of the WD in \cite{rubner2000earth}.
Terms $\mathbb{P}(F^z_n)$ and $\mathbb{P}(F^{z-1}_n)$ are the PMFs of \eqref{Num} at epochs $z$ and $z-1$ of the second phase, respectively.
\end{lemma}

\textbf{Thompson Sampling EE Phase}:
The last phase has a length of $\nu_3 2^z$ rounds where $\nu_3>0$ is a constant.
The objective is to find the best SF for each player based on the busy or idle state of the determined RIS in the second phase.
\com{Note that there are two types of resource allocations in this phase for transmission Pattern \RNum{1} and Pattern \RNum{2}.
The main difference between the two allocations is that,
in Pattern \RNum{1}, it requires jointly estimate $\theta^n_{k,c_m}$  and $\theta^n_{c_m}$ at each epoch;
while in Pattern \RNum{2}, it only needs to estimate the $\theta^n_{c_m}$ at each time slot.}

As mentioned before, the second-stage MPMAB problem can be regarded as a single-player MAB problem.
Therefore, we can adopt the TS algorithm \cite{agrawal2013further} to solve the second-stage MPMAB problem to track the Bernoulli distribution rewards (i.e., transmission success or failure).
The TS algorithm first maintains a Beta prior distribution\footnote{Beta$(\alpha,\beta)$ is the beta distribution with probability density function (pdf): $f_{\alpha,\beta}(y)=\frac{y^{\alpha-1}(1-y)^{\beta-1}}{B'(\alpha,\beta)},y\in [0,1], \ \mathrm{where} \ B'(\alpha,\beta)=\frac{\Gamma(\alpha)\Gamma(\beta)}{\Gamma(\alpha+\beta)}.$} for each SF.
Thus, the objective of this phase is equivalent to estimating the parameter in the Beta distribution, which will converge to the true value of $\theta^n_{k,c_m}$ or $\theta^n_{c_m}$.
Based on the transmission feedback, TS algorithm is able to update the posterior distribution by: $\alpha = \alpha + 1$ if transmission is successful, otherwise $\beta = \beta + 1$.
Notice that the value function (i.e., $c_m \hat{\theta}_{n,m}$) is the current data rate, instead of the successful or failed transmission feedback.
At the end of this phase, it can determine the current best SF for each IoT device by using the \com{estimated parameters in the Beta distribution}, i.e.,
\begin{equation}
 c_n^{\ast} = \arg  \max\limits_{m \in \mathcal{M}} \  \frac{c_m \alpha_{n,m}}{\alpha_{n,m}+\beta_{n,m}}.
\end{equation}

\subsection{Complexity and Feasibility Analysis }\label{PA1-2}
We first give a brief discussion on the computational complexity of the proposed algorithm.
In Algorithm \ref{GoT}, the computational complexity of the first phase is $\mathcal{O} (\nu_1 z^{\delta} L_{\mathrm{ED}})$, where $L_{\mathrm{ED}}$
is the length of samples in the energy detector of the spectrum sensing operation.
Meanwhile, the computational complexity of the second phase is $\mathcal{O} (\nu_2 z^{\delta} + z K\log_2K)$, where the second term comes from the WD in the calculation of parameter $\epsilon$ \cite{arras2019bound}.
In addition, the computational complexity of the third phase is $\mathcal{O} (\nu_3 2^z M\log_2 M)$, where the complexity comes from the `argument maximum' operation in the TS algorithm \cite{agrawal2013further}.
Therefore, the total computational complexity is $\mathcal{O} (\nu_1 z^{\delta} L_{\mathrm{ED}} + \nu_2 z^{\delta} + z K\log_2K + \nu_3 2^z M\log_2 M)$,
which increases linearly logarithmically with the number of RISs $K$ and SFs $M$.
As the time epoch $z$ increases, the complexity of the third phase will become the dominant factor in the total complexity.
Therefore, the total complexity of Algorithm 1 is about $\mathcal{O} (\nu_3 2^z M\log_2 M)$.

Next, we discuss the feasibility of the proposed algorithm in practical applications, e.g., B5G/6G networks.
First, the proposed algorithm performs in real time and automatically converges to the optimal solution (i.e., the online learning feature).
Second, its distributed feature can reduce the communication overhead and make it easy to apply to other network scenarios.
Third, the complexity of the proposed algorithm increases linearly logarithmically with the number of RISs $K$ and the SFs $M$.
These features demonstrate that the proposed algorithm has great potential for application in B5G/6G networks with different requirements for rate, delay, scalability, and reliability.

\com{In Algorithm \ref{GoT}, we only consider the case that the number of IoT devices is less than the RISs, i.e., $K\geq N$.}
However, the proposed algorithm can also handle the case of $K<N$ by dividing \com{$N$ IoT devices} into $K$ clusters using the $k$-means clustering method \cite{likas2003global} according to their geographic locations.
We assume that the IoT devices in the same cluster prefer the same RIS and communicate with the BS using the round-robin method.
Therefore, one of the clustering IoT devices can transmit on the RIS;
while others directly transmit data to the BS at each time slot.
In other words, Algorithm \ref{GoT} still works in the case of $K<N$ by allocating the optimal RIS to each cluster rather than each IoT device.

\begin{algorithm}[h]
\caption{Modified E2Boost Algorithm Run by Player $n$ for the case $K<N$}
\label{GoT2}
\begin{algorithmic}[1]
\State \textbf{Initialize:} Parameters in Algorithm \ref{GoT}
\For {each time slot $t=1,2,\cdots$, $T$}
\State Check the round-robin flag
\State If the flag is equal to $1$, run the E2Boost algorithm in Algorithm \ref{GoT}
\State Otherwise, run the TS algorithm in the third phase of Algorithm \ref{GoT}
\EndFor
\end{algorithmic}
\end{algorithm}
Therefore, we give a modified E2Boost algorithm to handle the case of $K<N$, as shown in Algorithm \ref{GoT2}.
At the beginning of each time slot,
\com{the IoT device $n$ first checks its round-robin flag to determine its transmission patterns:}
If the flag is equal to $1$, the IoT device runs the E2Boost algorithm in Algorithm \ref{GoT} to find the optimal RIS and SF;
otherwise, it runs the TS algorithm in the third phase of Algorithm \ref{GoT} to find the optimal SF.

\subsection{An Upper Pseudo-Regret Bound}\label{PA2}
We derive an upper pseudo-regret bound for the E2Boost algorithm.
Since each IoT device has two transmission patterns, the pseudo-regret also consists of the RIS-enabled regret and the non-RIS-enabled regret parts.

\com{For the RIS-enabled regret part, the performance analysis of the RIS-enabled regret is mainly built on \cite{bistritz2018game} since the E2Boost algorithm shares the same architecture of the GoT algorithm.
On the other hand, the Bernoulli reward processing in this work strictly meets the key condition of \emph{Definition 1} in \cite{bistritz2018game}.}
However, compared with the GoT algorithm, the proposed algorithm has the following features.
First, it is a two-stage MPMAB framework that has a small arm space (i.e., $\mathcal{K}$) to explore.
Its total pseudo-regret only depends on the number of RISs, instead of the whole arm space (i.e., $\mathcal{K} \otimes \mathcal{M}$) as that in the GoT algorithm.
Second, we embed the $\epsilon$-greedy algorithm into the first phase to further trade off the EE dilemma.
Thus, the accumulated regrets of this phase will trend to $0$  when all players agree on their optimal RISs.
Third, we incorporate the TS algorithm into the third phase to determine the best SF.
Similarly, only a few accumulated regrets will be accrued in this phase when the optimal RIS is determined.
\com{Therefore, these features} enable us to achieve a tighter pseudo-regret bound than the GoT algorithm.

\com{For the non-RIS-assisted regret part,} the E2Boost algorithm only has the third phase, i.e., the second-stage MPMAB problem.
Since two or more players that select the same arm (or SF) will not collide, this MPMAB problem is reviewed as a single stochastic MAB problem.
In Algorithm \ref{GoT}, we use the modified TS (MTS) algorithm \cite{Gupta2018Low-complexity} to solve this single stochastic MAB problem.
Therefore, we adopt the theoretical results in \cite{Gupta2018Low-complexity} to derive the non-RIS-enabled regret.

To conclude, we have the following theorem.
\begin{theorem}\label{Theorem01}
Let $\Gamma_{\max} = \max_{n,i} \mu_{n,i}$ be the maximum real expected rewards among all players' arms.
For any hybrid uplink network, given $\nu_1>0, \nu_2>0, \nu_3 >0$, $\delta\geq0$, $0<\varpi<1$ and a small enough $\varepsilon$,
the total upper pseudo-regret bound obtained by the E2Boost algorithm is
\begin{equation}
\begin{split}
\overline{Reg} \leq &N\Gamma_{\max} (1-P_a) \left( 2(\nu_1 +\nu_2) \log^{1+\delta}_2 \left(\frac{T}{\nu_3}+2\right)\right.\\
&\left.+(6NK+1)\nu_3 \log_2 \left(\frac{T}{\nu_3}+2\right)\right) \\
&+  P_a (1+\varpi) \sum_{n=1}^{N} \sum_{a_n\in \mathcal{M}} \frac{\log_2 T}{\mathrm{D}_\mathrm{KL}(a_n, a_n^{\ast})} \Delta_{n,a_n},
\end{split}
\end{equation}
where $\log_2^{1+\delta} \left(T/\nu_3+2 \right)$ denotes $\log_2 \left(T/\nu_3+2 \right)$ to the power of $(1+\delta)$ and $\mathrm{D}_\mathrm{KL}(\cdot)$ is the Kukkback-Leibler divergence.
Term $P_a$ is the active probability of the UE.
\end{theorem}
\begin{proof}
See Appendix \ref{appendix1}.
\end{proof}

\begin{remark}\label{remark1}
The first two terms of the upper pseudo-regret bound account for the RIS-assisted regret part;
while the third term is the non-RIS-assisted regret part.
We can see that the weights of these two parts rely on the active probability of the UE.
\end{remark}

\begin{remark}\label{remark2}
The total upper pseudo-regret bound increases logarithmically with $T$, i.e., $\overline{Reg} = \mathcal{O}(\log^{1+\delta}_2 T)$,
indicating that Algorithm \ref{GoT} will converge and the per-round regret approaches zero when $T$ is sufficiently large.
\end{remark}

\begin{remark}\label{remark3}
The total upper pseudo-regret bound in the E2Boost algorithm is much tighter than the GoT algorithm.
\com{According to \cite{bistritz2018game}, the total upper pseudo-regret bound of the GoT algorithm is}
\begin{equation}
\begin{split}
 \overline{Reg}_{\mathrm{GoT}} \leq &4N\Gamma_{\max}(\nu_1+\nu_2) \log^{1+\delta}_2 \left(\frac{T}{\nu_3}+2\right) \\
& + N\Gamma_{\max}(6NKM+1)\nu_3 \log_2 \left(\frac{T}{\nu_3}+2\right) \\
&= \mathcal{O}(\log^{1+\delta}_2 T).
\end{split}
\end{equation}
For example, when $\nu_1=\nu_2=\nu_3$, $\delta = 0$, \com{and $P_a=0$,} we have
\begin{equation}
\begin{split}
\overline{Reg} \leq \left(5+6NK \right)\nu_1 N\Gamma_{\max}  \log_2 \left(\frac{T}{\nu_3}+2\right),
\end{split}
\end{equation}
and
\begin{equation}
\begin{split}
\overline{Reg}_{\mathrm{GoT}} \leq \left(9+6NKM \right)\nu_1 N\Gamma_{\max}  \log_2 \left(\frac{T}{\nu_3}+2\right).
\end{split}
\end{equation}
It can be seen that $\overline{Reg}$ is about $M$ times lower than $\overline{Reg}_{\mathrm{GoT}}$.
This observation can be verified by the numerical results in the following section.
\end{remark}

\section{Simulation Results}\label{NA2}
We conduct extensive simulations to evaluate the performance of the proposed algorithms.
The simulation parameters are chosen according to the 3GPP standard \cite{3GPP} and refs. \cite{zhang2020reconfigurable, di2019hybrid}.
All results are obtained from $10^3$ Monte Carlo (MC) trials.

\subsection{Parameter Configuration and Baseline Algorithms}
\textbf{Parameter Configuration:}
The transmit power at each IoT device is $\Omega_n = 20$ dBm,  $\forall n \in \mathcal{N}$.
The background noise plus interference power is $ -95$ dBm, and the wavelength $\lambda$ is set according to the central carrier frequency $5.9$ GHz.
The bandwidth $B$ is $40$ MHz.
The Rician factor is $\zeta = 4$, and the antenna gain $G$ is set to $1$.
Each IoT device has $6$ SFs to choose from, as shown in TABLE \ref{Rate}.
The data rates are determined by \eqref{SF} and the thresholds are the minimum required SINR to demodulate the received signal.
Assume that the active probability of each RIS (i.e., occupied by the legal UEs) is $P^k_{a}=0.2$.
In addition, we adopt the UMa model \cite{3GPP} to describe the path loss of both LoS and NLoS components.
\begin{table}[!t]
\renewcommand{\arraystretch}{1.3}
\caption{The transmission parameters for the C-IoT device}
\label{Rate}
\centering
\begin{tabular}{|c|c|c|c|c|c|c|}
  \hline
  SF   & $7$ & $8$ & $9$ & $10$ & $11$ & $12$ \\
  \hline
  Data Rate (Mbps) & $1.09$ & $0.63$ & $0.35$  & $0.20$ & $0.11$ & $0.06$ \\
  \hline
  Threshold ($\times 10^3$) & $4.5$ & $4$ & $3.5$ & $3$ & $2.5$ & $2$ \\
  \hline
\end{tabular}
\end{table}

The RIS is placed perpendicular to the ground, and the number of elements is $101 \times 101$.
The direction of the RIS in the $XY$-plane is shown in Fig. \ref{RS}.
The angle $\angle \varphi$ and all elements' locations in RIS are determined according to Appendix \ref{appendix2}.
Each element contains $b = 8$ PIN diodes with the refection amplitude $A = 1$.
We consider two types of phase shift settings, i.e., the \emph{optimal phase shift} and the \emph{constant phase shift}.
For the \emph{optimal phase shift} setting, each RIS's phase shifts are set to be optimal to the UEs (\emph{see Proposition 2} of \cite{zhang2020reconfigurable}), i.e.,
\begin{equation}
\begin{split}
\tau_{l_1, l_2} = \left\lfloor \left(\Pi - \frac{2\pi}{\lambda}L^k_{l_1,l_2} \right)\frac{2^b}{2 \pi} \right\rfloor \frac{2\pi}{2^b},
\end{split}
\end{equation}
where $\Pi$ is an arbitrary constant and $L^k_{l_1,l_2}$ is the distance between the BS, and the UEs through the $k$-th RIS's $(l_1, l_2)$ element.
For the \emph{constant phase shift} setting, we assume that the phase shifts on all RISs' elements are equal.
That is, all integers $\rho_{l_1, l_2}$ in \eqref{phase-shift} are simply set to a constant (we set to $170$ in the following simulations).
Note that $\rho_{l_1, l_2}$ can be an arbitrary integer in the range of $[0, 2^b-1]$.
\begin{figure}[!t]
\centering
\includegraphics[width=2.3in]{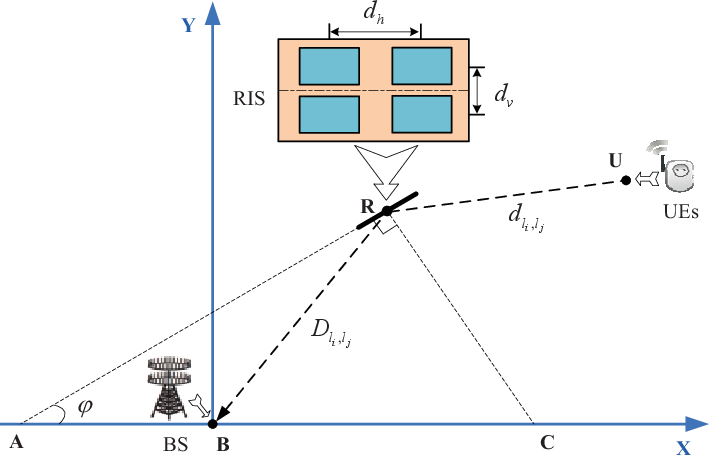}
\caption{A schematic diagram of the UE-RIS-BS link  (top view).}
\label{RS}
\end{figure}

\textbf{Baseline Algorithms:}
\com{We compare the E2Boost algorithm with the optimal solution, GoT algorithm, Q-learning method, random selection method, E2Boost without TS algorithm, and E2Boost without WD algorithm.
Next, we introduce these baseline algorithms in detail.}
\begin{itemize}
  \item \emph{Optimal Solution:}
  The optimal solution is obtained by solving the two-stage MPMAB problem in a centralized form.
  \com{In the Pattern \RNum{1}, it allocates the optimal RIS and SF to each IoT device by using the Hungarian algorithm \cite{papadimitriou1998combinatorial},
  where the only required information is $\theta_{k,c_m}^n$ and $\theta_{c_m}^n$.}
  In the simulation, we obtain this information by recording the received SINR $\gamma_{n,k}$ and $\gamma_{n}$ with the above simulation parameters over $10^5$ MC trails.
  Then, $\hat{\theta}_{k,c_m}^n$ and $\hat{\theta}_{c_m}^n$ can be estimated by comparing these SINRs with a given threshold $\Psi_m$.
  Note that $\hat{\theta}_{k,c_m}^n$ and $\hat{\theta}_{c_m}^n$ can approach the true values of $\theta_{k,c_m}^n$ and $\theta_{c_m}^n$ arbitrarily as long as the number of MC trials is sufficiently large.
  Based on this information and the data rate $c_m$ in Table \RNum{1}, \com{the optimal RIS and SF of each IoT device} can be obtained by using the Hungarian algorithm (i.e., the \emph{munkres} function in Matlab).
  \com{In the Pattern \RNum{2},} we determine the optimal SF for each IoT device using the genie-aided solution (i.e., from God's perspective) as the $\hat{\theta}_{c_m}^n$ and $c_m$ are known.
  \item \emph{GoT Algorithm:} The GoT algorithm in \cite{bistritz2018game} is a fully distributed algorithm to solve the decentralized resource allocation problems.
  \com{It has the same architecture as the proposed algorithm.}
  However, it lacks the $\epsilon$-greedy algorithm and the TS algorithm in the first and third phases to further balance the EE dilemma.
  In addition, it needs to explore the combinations of RISs and SFs; while the proposed algorithm explores the RISs and SFs separately.
  \item \emph{Q-learning method:}
  For the Q-learning method, \com{the state is the target RIS's busy or idle state.}
  The state transition probability is the RIS's active or passive probability $P^k_a$ or $1-P^k_a$.
  The actions are the set of SFs $\mathcal{M}$ if the target RIS is in a busy state;
  otherwise, the actions are the combinations of SFs and RISs, i.e., $\mathcal{K}\otimes \mathcal{M}$.
  \item \emph{Random Selection:} For the random selection method, each IoT device uniformly chooses an arm from the arm space $\mathcal{K}\otimes \mathcal{M}$ in Pattern
  \RNum{1} or the arm space $\mathcal{M}$ in Pattern \RNum{2} at each time slot. There is no EE mechanism inside.
  \item \emph{E2Boost without TS:} Compared with the E2Boost algorithm, it removes the TS algorithm from the third phase. Moreover, it requires exploring the combinations of the RISs and SFs in the first phase with the $\epsilon$-greedy algorithm.
  \item \emph{E2Boost without WD:} Compared with the E2Boost algorithm, the only difference is that it maintains a constant exploration rate $\epsilon$ for the $\epsilon$-greedy algorithm, \com{rather than adaptively adjusting $\epsilon$ in the E2Boost algorithm.}
\end{itemize}

It is worth noting that there is a performance gap between the solutions of the two-stage MPMAB problem and the original problem \eqref{DSO}.
Specifically, the solution of problem \eqref{DSO} is to allocate the optimal available RIS and SF to each IoT device at each time slot $t$;
while the solution (i.e., the optimal solution) of the two-stage MPMAB problem is to assign the optimal RIS and SF to each IoT device average over the time horizon $T$.
As a result, the performance of the two-stage MPMAB problem is slightly \com{poorer} than that of problem \eqref{DSO}.
However, Theorem \ref{Theorem01} shows that the proposed algorithm can converge to the optimal solution when $T$ is sufficiently large.
The following simulation results will also verify this.

\subsection{Fixed Network Scenario}
We first consider a fixed network scenario in a $200$ $\mathrm{m}$ $\times$ $200$ $\mathrm{m}$ square area, as shown in Fig. \ref{NS},
where $N = 3$ cellular IoT devices are located in a $45$ $\mathrm{m}$ $\times$ $45$ $\mathrm{m}$ circle area.
For simplicity, we assume that all UEs are centered in the point $(x,y)=(150, 150)$ $\mathrm{m}$.
Outside this circle are the BS and $K=3$ RISs.
The distances between the BS and the center of the RIS, as well as the IoT device and the center of the RIS, are calculated by the Euclidean distances \emph{w.r.t.} their locations (i.e., $D_{l_i,l_j}$ and $d_{l_i,l_j}$ in Fig. \ref{RS}), respectively.
The RIS and the BS heights are $10$ m and $20$ m, respectively.
\begin{figure}[!t]
\centering
\includegraphics[width=2.2in]{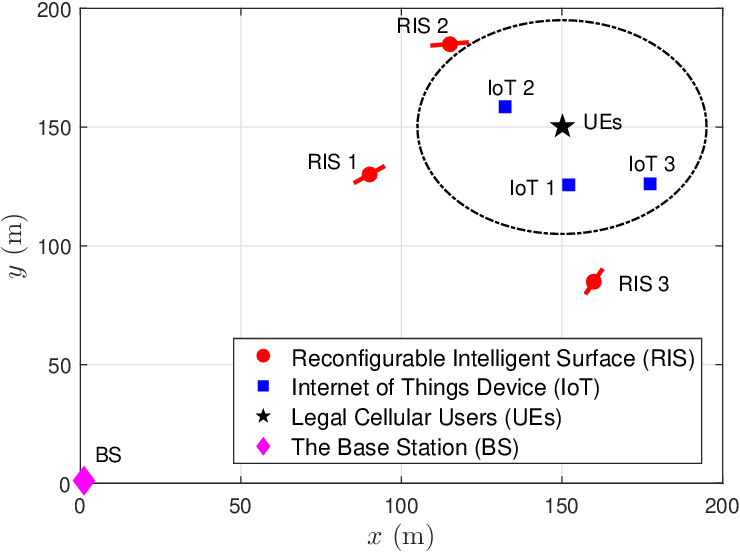}
\caption{A fixed network scenario in a $200\times 200$ m square area with $K=3, N=3$ (top view).}
\label{NS}
\end{figure}

Fig. \ref{MultiRate_NumArm01} shows the allocation results of the E2Boost algorithm for the \emph{optimal phase shift} setting.
The simulation parameters for the E2Boost algorithm are $\nu = 1.4$, $\delta = 0$, $\varepsilon = 0.01$, $Z = 10$, $\nu_1 = \nu_2 = 1000$ and $\nu_3 = 100$.
The average throughput is computed by $\frac{1}{t} \sum_{t=1}^T \mu_{n, I_{n,t}}$.
It can be seen from Fig. \ref{MultiRate_NumArm01} (b-d) that each player will converge to its own optimal SF and RIS,
i.e., player 1 $\rightarrow$ (RIS3, SF1), player 2 $\rightarrow$ (RIS1, SF1) and player 3 $\rightarrow$ (RIS2, SF1).
All players prefer SF1 with the highest data rate of $1.09$ Mbps in Table \ref{Rate}.
The average throughput of all players is $2.3859$ Mbps, which is slightly less than the optimal solution's $2.4315$ Mbps.
In addition, IoT device 3 accounts for the lowest average throughput by $0.4971$ Mbps due to the long transmission distance between IoT 3-RIS2-BS links.
The IoT device 3 does not choose RIS3 because the direction (or the phase shifts) of RIS2 is more suitable for IoT 3 than RIS3, i.e., RIS2-UEs-IoT3 in a line.

Similarly, Fig. \ref{MultiRate_NumArm02} shows the allocation results of the E2Boost algorithm for the \emph{constant phase shift} setting.
We can see that the average throughput of all players is just about $0.5782$, which is much lower than the \emph{optimal phase shift} setting.
In addition, player 1 and player 2 disagree on the optimal RIS since there is a collision between them.
The reason is that the time horizon of the second phase in Algorithm \ref{GoT} is too short (i.e., $\nu_2=1,000$) to resolve this collision.
As a result, the highest SF for Pattern \RNum{2} is chosen frequently, resulting in low average throughput.
To conclude, Figs. \ref{MultiRate_NumArm01} and \ref{MultiRate_NumArm02} demonstrate that the channel gains between the IoT device and the BS not only rely on the path-loss gain but also depend on the settings of phase shifts and direction in the RIS.

\setlength{\textfloatsep}{5pt}
\begin{figure*}[!t]
\centering
\subfloat[Throughput]{\includegraphics[width=1.3in]{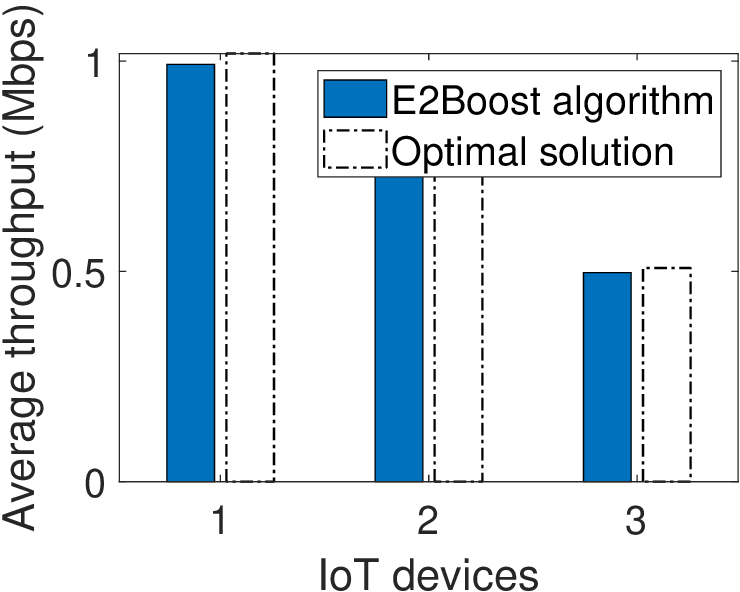}%
\label{TO}}
\hfil
\subfloat[IoT device 1 (Player 1)]{\includegraphics[width=1.4in]{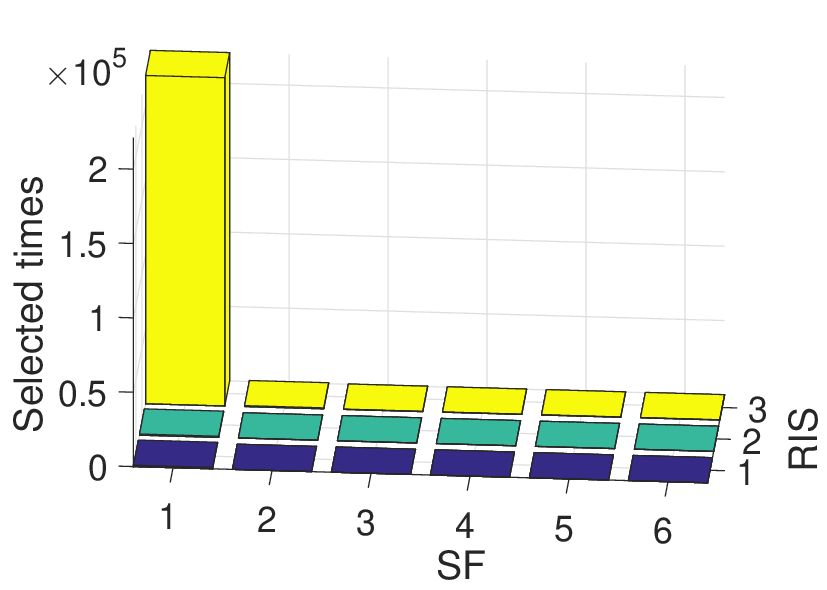}%
\label{SO1}}
\hfil
\subfloat[IoT device 2 (Player 2)]{\includegraphics[width=1.4in]{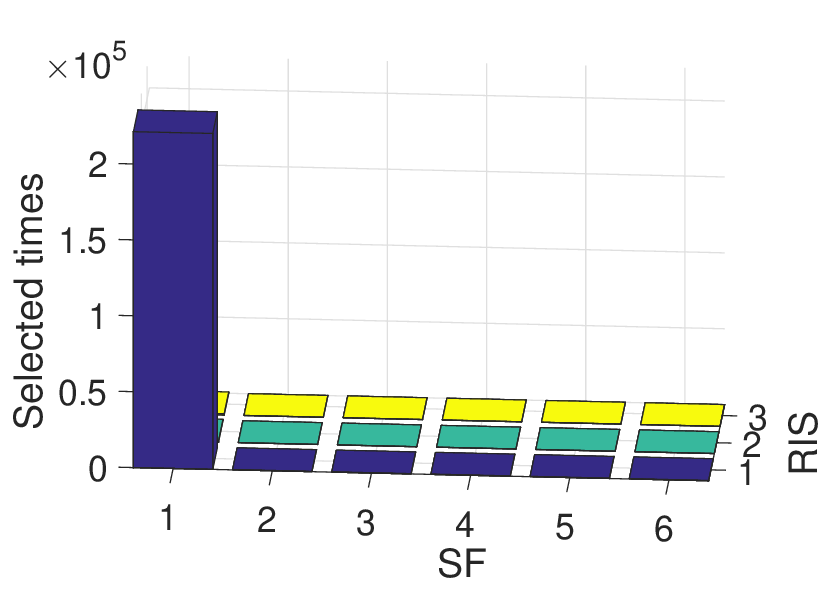}%
\label{SO2}}
\hfil
\subfloat[IoT device 3 (Player 3)]{\includegraphics[width=1.4in]{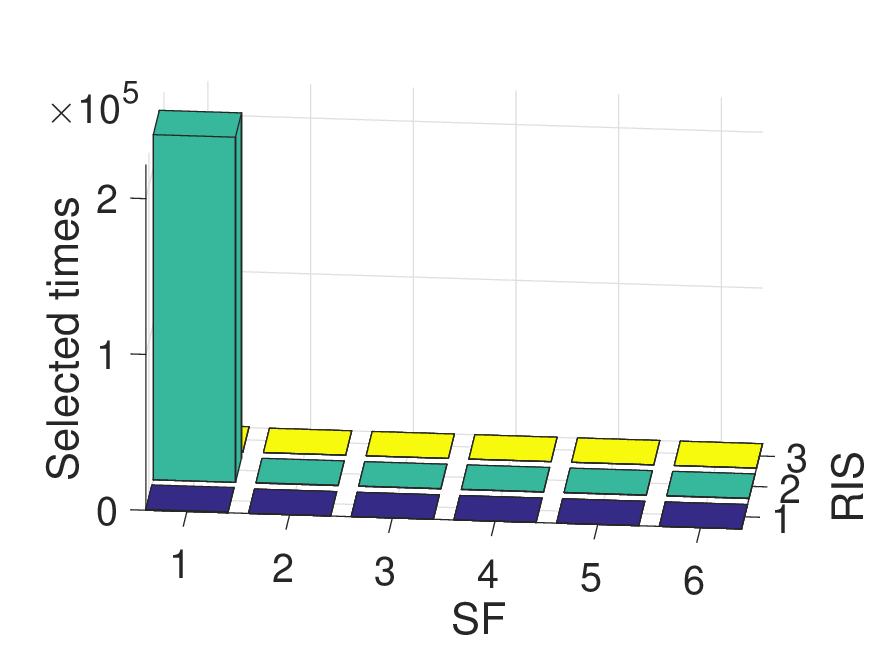}%
\label{SO3}}
\caption{(a) Average throughput of three IoT devices, (b-d) The number of selected times at each arm for each player, by running the E2Boost algorithm with \emph{optimal phase shift} setting in the fixed network scenario Fig. \ref{NS}.}
\label{MultiRate_NumArm01}
\end{figure*}
\setlength{\textfloatsep}{5pt}
\begin{figure*}[!t]
\centering
\subfloat[Throughput]{\includegraphics[width=1.3in]{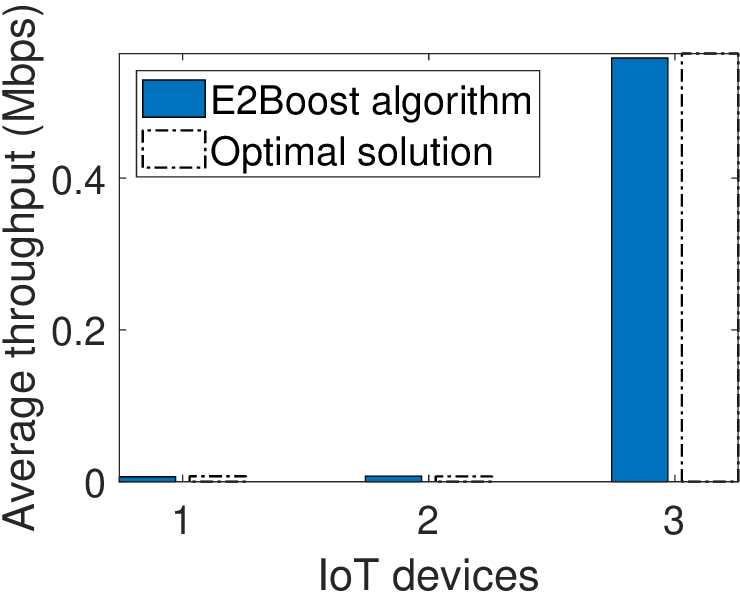}%
\label{TR}}
\hfil
\subfloat[IoT device 1 (Player 1)]{\includegraphics[width=1.4in]{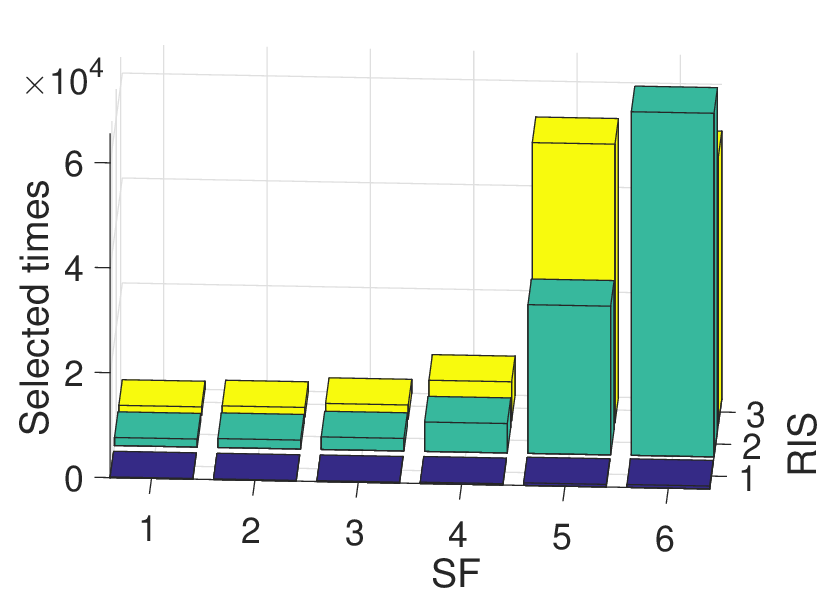}%
\label{SR1}}
\hfil
\subfloat[IoT device 2 (Player 2)]{\includegraphics[width=1.4in]{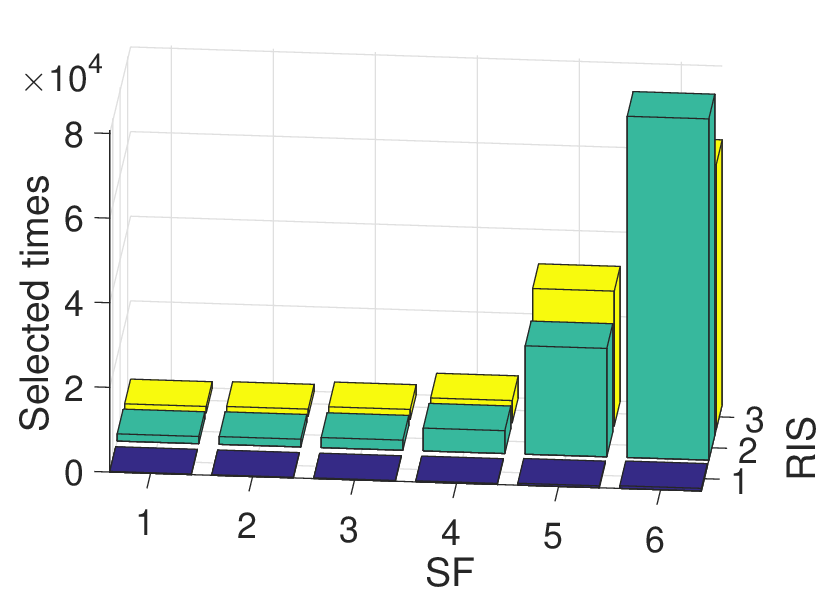}%
\label{SR2}}
\hfil
\subfloat[IoT device 3 (Player 3)]{\includegraphics[width=1.4in]{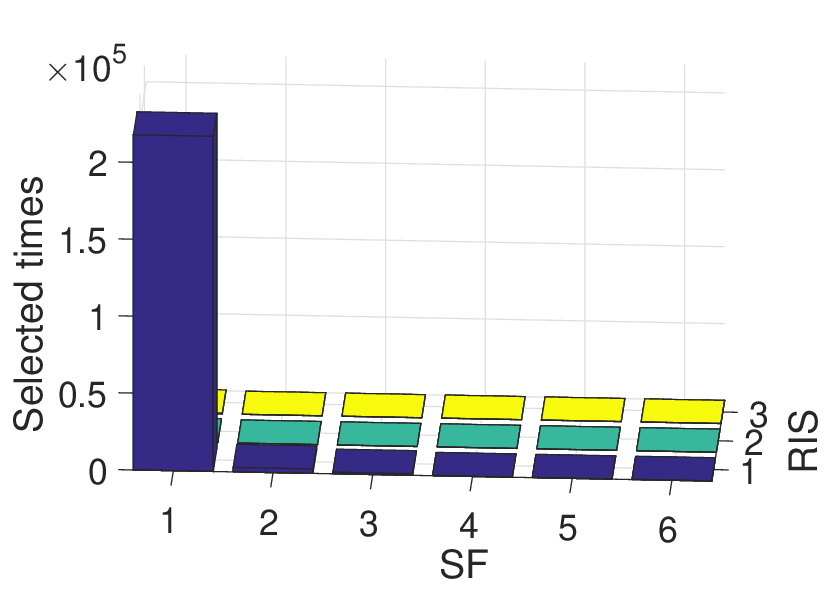}%
\label{SR3}}
\caption{(a) Average throughput of three IoT devices, (b-d) The number of selected times at each arm for each player, by running the E2Boost algorithm with \emph{constant phase shift} setting in the fixed network scenario Fig. \ref{NS}.}
\label{MultiRate_NumArm02}
\end{figure*}

Fig. \ref{MultiRate_Regret} depicts the total pseudo-regret of the E2Boost algorithm, the E2Boost algorithm without TS, and the GoT algorithm in the cases of $\nu_1 = \nu_2 = 1,000$ and  $\nu_1 = \nu_2 = 2,000$, when $Z = 10$ under the \emph{optimal phase shift} setting.
Other parameters are the same as those in Fig. \ref{MultiRate_NumArm01}.
We can see that the proposed algorithm has the lowest expected total pseudo-regret in both cases since
it has a small arm space (i.e., due to the two-stage allocation mechanism) to explore.
In addition, the total pseudo-regrets of all algorithms in the case of $\nu_1 = \nu_2 = 1,000$ are lower than those in the $\nu_1 = \nu_2 = 2,000$ case.
The reason is that a larger value of $\nu_1$ and $\nu_2$ indicates that a longer time is needed to explore all arms, leading to more performance loss.
However, when $\nu_1 = \nu_2 = 1,000$, the GoT algorithm and the E2Boost algorithm without TS will not converge since the value of $\nu_2$ is too small for the second phase to resolve the collisions among IoT devices.
More importantly, Fig. \ref{MultiRate_Regret} validates our theoretical analysis of Theorem \ref{Theorem01}, where the total pseudo-regret of the E2Boost algorithm increases logarithmically \emph{w.r.t.} the time horizon $T$ and is about four times better than the GoT algorithm.
\begin{figure}[!t]
\centering
\includegraphics[width=2.2in]{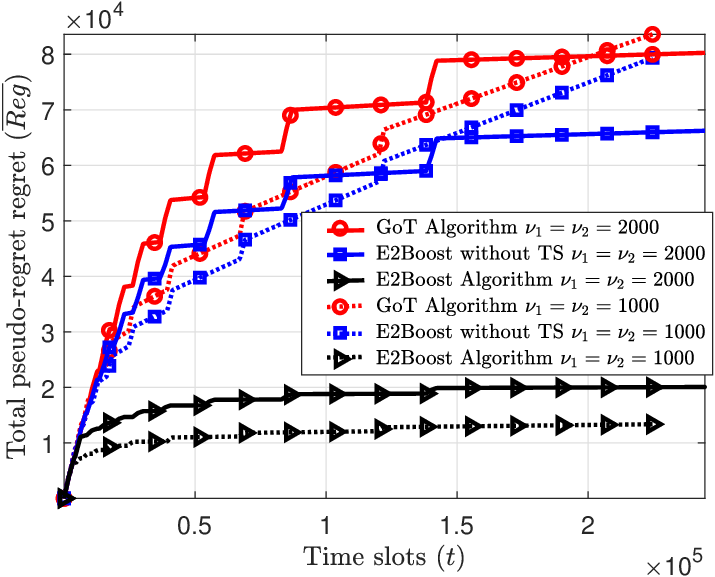}
\caption{The total pseudo-regret via time slot in the cases of $\nu_1 =\nu_2 = 1,000$ and  $\nu_1 = \nu_2 = 2,000$ with \emph{optimal phase shift} setting in the fixed network scenario Fig. \ref{NS}.}
\label{MultiRate_Regret}
\end{figure}

Fig. \ref{MultiRate_Throughput} compares the average total throughput of the E2Boost algorithm, the E2Boost algorithm with $\epsilon=0$ and $\epsilon=1$ (without WD), the E2Boost algorithm without TS, the GoT algorithm, and the random selection method in the \emph{optimal phase shift setting} with $\nu_1 = \nu_2 = 2,000, Z = 10$.
It can be seen that the proposed algorithm outperforms the other algorithms and is close to the optimal solution.
By contrast, the proposed algorithm with $\epsilon=0$ accounts for the worst performance since there is no exploration in the first phase.
Meanwhile, the E2Boost algorithm with $\epsilon=1$ and the E2Boost algorithm without TS is better than the GoT algorithm,
indicating that the E2Boost algorithm with WD can effectively trade off the EE dilemma by sequentially optimizing the parameter $\epsilon$.
More importantly, since the two-stage allocation mechanism, the E2Boost algorithm has a faster convergence rate than the GoT algorithm and the E2Boost algorithm without TS.
\begin{figure}[!t]
\centering
\includegraphics[width=2.2in]{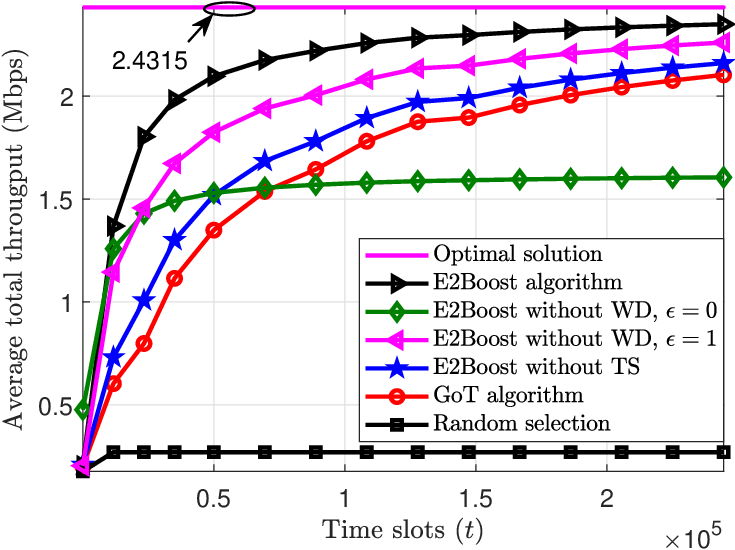}
\caption{The performance of different algorithms versus time slot with \emph{optimal phase shift} setting when $\nu_1 = \nu_2 = 2,000, Z=10$ in a fixed network scenario of Fig. \ref{NS}.}
\label{MultiRate_Throughput}
\end{figure}

Next, we evaluate the impact of the RIS-enabled channel on the performance of the proposed algorithm.
Fig. \ref{PhaseShift_RiceFactor} depicts the performance of the E2Boost algorithm under the optimal and constant phase shift setting for different Rice factors ($\zeta = 0.5, 1, 4, 10$) when $\nu_1 = \nu_2 = 2,000, Z=10$.
We can see that the performance of the \emph{optimal phase shift} setting is much better than the \emph{constant phase shift} setting for different Rice factors.
This is because the optimal phase shift is designed for the centralized UEs.
Thus, IoT devices close to the UEs will also have better performance.
On the other hand, a bigger $\zeta$ will result in a higher average total throughput.
This phenomenon can be explained by \eqref{ChannelModel}, where a big $\zeta$ means that the channel gain is dominated by the LoS component, i.e., the directional reflection link of IoT-RIS-BS.
Therefore, the channel gain is dominated by the RIS when $\zeta$ trends to $+\infty$;
while $\zeta$ trends to $0$ mean that the IoT device only transmits on Pattern \RNum{2}.
\begin{figure}[!t]
\centering
\includegraphics[width=2.2in]{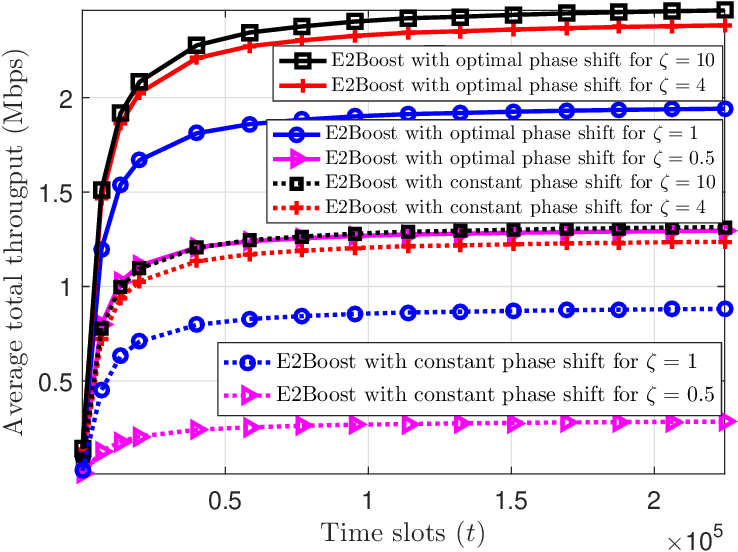}
\caption{The performance of E2Boost algorithm for different phase shift settings and Rice factors in a fixed network scenario of Fig. \ref{NS}.}
\label{PhaseShift_RiceFactor}
\end{figure}

\subsection{Random Network Scenario}
In the following, we evaluate the proposed algorithm under the random network scenario.
At each MC trial, we regenerate the locations of the IoT devices uniformly in the circle area of Fig. \ref{NS}.
Meanwhile, the distance of any two devices is subject to no less than $5$ m.
The locations of RISs, UEs, and BS are set the same as those in Fig. \ref{NS}.

Fig. \ref{Random_AveThro} compares the average total throughput of different algorithms in the\emph{ optimal phase shift} setting with $\nu_1 = \nu_2 = 2,000, Z=11$ over $10^3$ random network scenarios.
It can be observed that the performance of all algorithms except the random selection method increases with time slot $t$.
Again, the E2Boost algorithm has the best performance and a fast convergence rate.
The Q-learning method also exhibits a fast convergence rate, but it suffers from some performance loss due to the lack of the non-cooperation game phase to resolve the collisions among players.
Moreover, the gaps between the optimal solution and these algorithms \com{increase} compared with Fig. \ref{MultiRate_Throughput} in the fixed network case.
The reason is that these algorithms fail to find the optimal RIS for each player under some extreme network scenarios with the constant parameter $\nu_2$ and the limited time horizon $T$.
\begin{figure}[!t]
\centering
\includegraphics[width=2.2in]{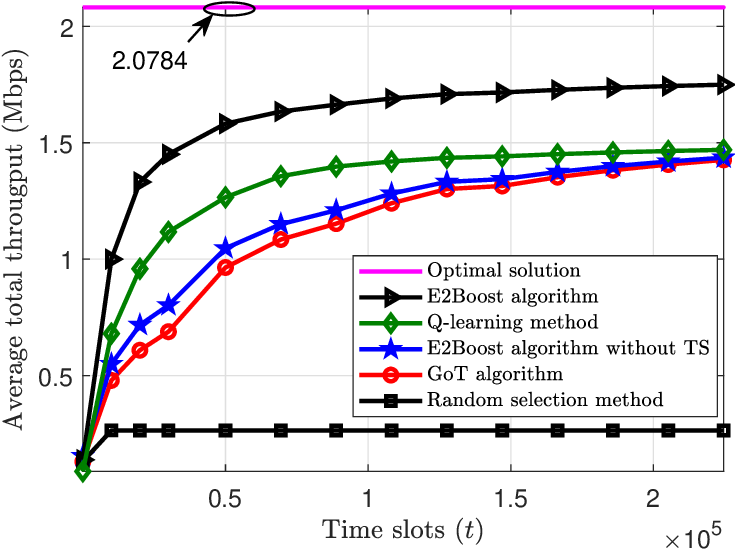}
\caption{The performance of different algorithms versus time slot with \emph{optimal phase shift} setting when $\nu_1 =\nu_2= 2000, Z=10$ over $10^3$ random network scenarios.}
\label{Random_AveThro}
\end{figure}

\begin{figure}[!t]
\centering
\subfloat[A random network scenario] {\label{NS2}
\includegraphics[width=0.47 \columnwidth]{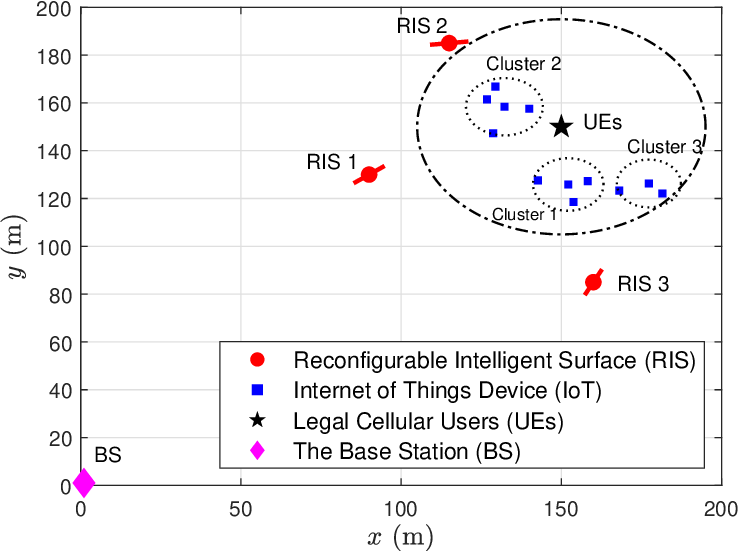}}   
\hfil
\subfloat[Performance comparison at left network scenario] { \label{Random_Cluster}   
\includegraphics[width=0.47 \columnwidth]{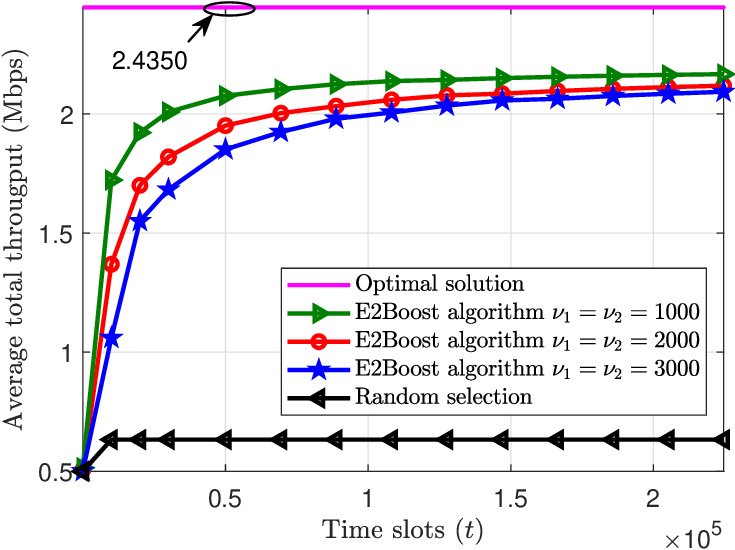}}   
\caption{(a) A random network scenario in a $200\times 200$ m square area with $K=3, N=11$.
(b) The performance of different algorithms versus time slots with \emph{optimal phase shift} setting over $10^3$ random network scenarios.}  
\label{FirstComb} 
\end{figure}
Moreover, we study the performance of the proposed algorithm by considering the case that the number of IoT devices is larger than that of RISs, i.e., $N>K$.
We first generate a new random network scenario,
as shown in Fig. \ref{NS2}, where $N=11$, $K=3$, and the other parameters are the same as those in Fig. \ref{NS}.
We can see from Fig. \ref{NS2} that $N=11$ IoT devices are divided into three clusters
by using the $k$-means clustering method according to their geographic locations.
Fig. \ref{Random_Cluster} compares the performance of the modified E2Boost algorithm (i.e., Algorithm \ref{GoT2}) with different settings of $\nu_1=\nu_2= \{1000, 2000, 3000\}$, and the random selection method in the \emph{optimal phase shift} setting over $10^3$ random network scenarios of  Fig. \ref{NS2}.
It can be seen that the modified E2Boost algorithm with $\nu_1=\nu_2= 1,000$ has the best performance,
and all the algorithms except for the random selection method can converge to the optimal allocation.
Compared with the results in Fig. \ref{Random_AveThro},
the average total throughput in the network scenario of Fig. \ref{NS2} is about $2.4350$ Mbps, which is slightly better than $2.0784$ Mbps in Fig. \ref{NS}.
This demonstrates that, although the number of IoT devices \com{increases},
the performance gain from the non-RIS-assisted transmission pattern is insignificant.

At last, we investigate the influence of the number of IoT devices on the proposed algorithm.
The number of RISs is set to $10$ and is placed on a semicircle with a radius of $55$ m from $3\pi/4$ to $5\pi/4$, as shown in Fig. \ref{NS_RIS10}.
The distances between two neighboring RISs are equal except for the two pairs that are located in the middle and both ends.
\begin{figure}[!t]
\centering
\subfloat[A random network scenario] {\label{NS_RIS10}
\includegraphics[width=0.47 \columnwidth]{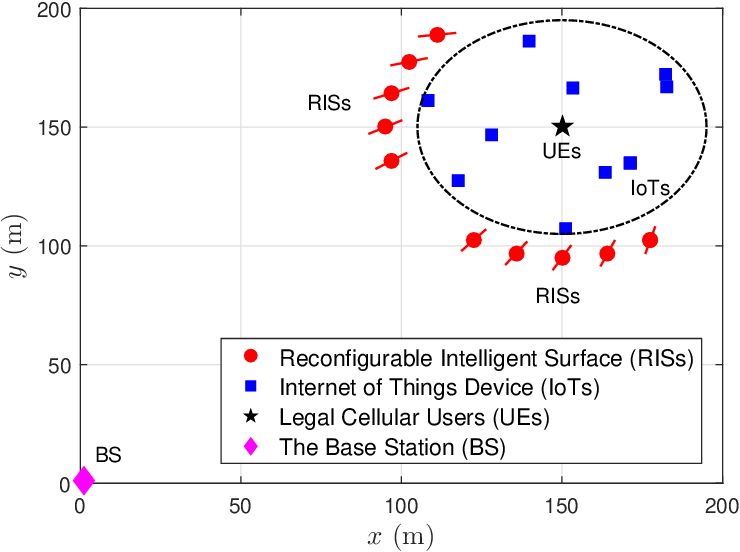}}   
\hfil
\subfloat[Performance comparison at left network scenario] { \label{NumRIS_AveThro}   
\includegraphics[width=0.47 \columnwidth]{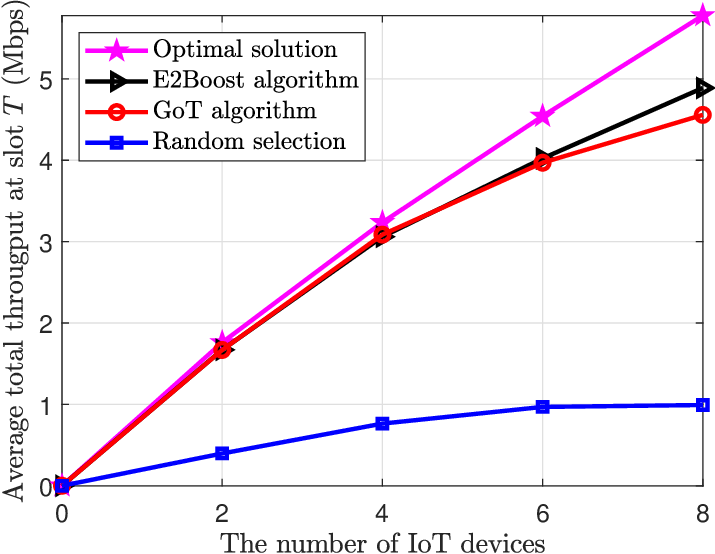}}   
\caption{(a) A random network scenario in a $200$ m $\times$ $200$ m square area with $K = 10$.
(b) The performance of the optimal solution, the E2Boost algorithm, the original GoT algorithm, and the random selection method versus the number of IoT devices at $10^3$ random network scenarios of the left figure.}  
\label{SecondComb} 
\end{figure}
Fig. \ref{NumRIS_AveThro} shows the performance of the optimal solution, the E2Boost algorithm, the original GoT algorithm, and the random selection method versus the number of IoT devices in the \emph{optimal phase shift} setting where $\nu_1 =  \nu_2 = 2,000, K=10, Z=10$ over $10^3$ random network scenarios of Fig. \ref{NS_RIS10}.
It can be seen that the performance of these algorithms increases with the number of IoT devices.
However, the proposed algorithm is better than the GoT algorithm and the random selection method since it has a small arm space to explore.
In addition, the performance gaps between the optimal solution and these algorithms increase with the number of IoT devices.
The reason is that \com{collision probabilities} among players increase with the number of IoT devices, resulting in more performance loss.

\section{Conclusion and Discussion}
This paper studied the resource allocation problem in a RIS-assisted hybrid uplink network where several IoT devices transmit data to the BS.
The objective is to maximize the sum rates of all IoT devices by finding the optimal RIS and SF for each device.
We modeled this problem as a two-stage MPMAB framework, \com{where the first stage is to find the optimal RIS, and the second stage is to find the optimal SF.}
Then, we proposed an E2Boost algorithm to tackle this problem by combining the $\epsilon$-greedy algorithm, TS algorithm, and non-cooperation game method.
Therefore, it can efficiently balance the EE dilemma.
Furthermore, we provided an upper regret bound for the \com{E2Boost} algorithm, i.e., $\mathcal{O}(\log^{1+\delta}_2 T)$,
indicating that the per-round regret will trend to $0$ when $T$ is sufficiently large.
In addition, simulation results demonstrated the effectiveness of the proposed algorithm.
More importantly, it is not sensitive to the joint arm space thanks to the two-stage allocation mechanism,
which can benefit practical applications.

\com{In the system model, we assume that different RISs use different frequencies, and one RIS can serve at most one IoT device.
A more general scenario is that RIS can reuse these frequencies and serve multi-IoT devices.
Then, two interesting problems are how to design a mechanism that allows the UE and IoT device signals to coexist in the same RIS and how to design the RIS-assisted multi-IoT system by estimating the exact information of the RIS and CSI.
These are important yet challenging problems for future study.}

\appendices
\section{Proof of Theorem 1}\label{appendix1}
At each time slot, IoT device \com{transmits on either the Pattern \RNum{1} or} the Pattern \RNum{2}.
For the Pattern \RNum{1}, the total pseudo-regret term $\overline{Reg}^{(1)}$ can be expanded by investigating $\overline{Reg}_z$, where $z$ is the epoch.
Thus, we begin to bound $\overline{Reg}_z$ by computing the probability of event $E_z$,
which is the event that the optimal assignment $\boldsymbol{a}^{\ast}$ is not adopted in \com{the third phase at epoch $z$}.
We have
\begin{equation}\label{AA1}\small
\begin{split}
\mathrm{Pr}(E_z) &= \mathrm{Pr} \left(E^{k^\ast}_z, E^{m^\ast}_z \right) + \mathrm{Pr} \left(E^{k^\ast}_z, \overline{E^{m^\ast}_z} \right) + \mathrm{Pr} \left(\overline{E^{k^\ast}_z}, E^{m^\ast}_z \right)\\
&=\mathrm{Pr} \left(E^{k^\ast}_z\right) + \mathrm{Pr} \left(\overline{E^{k^\ast}_z}, E^{m^\ast}_z \right),
\end{split}
\end{equation}
where $E^{k^\ast}_z$ is the event that the optimal RIS is not used at the end of the  $z$-th epoch of the second phase,
and $E^{m^\ast}_z$ is the event that the optimal SF is not used at the end of the $z$-th epoch of the third phase.

First, we bound the probability that event $E^{k^\ast}_z$ holds, i.e.,
\begin{equation}\label{AA2}
\mathrm{Pr} \left(E^{k^\ast}_z\right) \leq \mathrm{Pr} \left(\bigcup_{j=0}^{\lfloor \frac{z}{2}\rfloor} P_{e,z-j} \right) + P_{g,z},
\end{equation}
where $P_{e,z}$ is the probability that \com{the optimal assignment is different from $\boldsymbol{a}^{\ast}$ in the first phase at epoch $z$,}
and $P_{g,z}$ is the probability that \com{the frequently visited strategy profile is not $\boldsymbol{a}^{\ast}$ in the last $\lfloor z/2\rfloor +1$ non-cooperation game phases}.
Then, we need to calculate the probabilities of $P_{e,z}$ and $P_{g,z}$ before bounding $\overline{Reg}_z$.
In the first phase, we estimate the average successful transmission probabilities $\hat{\theta}_{n,k}$ of all RISs.
Assume i.i.d. rewards $X_{n,k}$ and each player uniformly explores all $K$ RISs when event $E^{k^\ast}_z$ holds.
By adopting the result in \cite{bistritz2018game} (see \emph{Lemma 8}), we have
\begin{equation}\label{A2}
P_{e,z} \leq 2NK e^{-w\nu_1\left(\frac{z}{2}\right)^\delta z} + NK e^{-\frac{\nu_1 \left(\frac{z}{2}\right)^\delta }{36K^2} z},
\end{equation}
where $w$ is a predefined positive constant.
Therefore,
\begin{equation}\label{A3}\small
\begin{split}
\mathrm{Pr} \left(\bigcup_{j=0}^{\lfloor \frac{z}{2}\rfloor} P_{e,z-j} \right)
\leq \frac{2NK e^{-\frac{w}{2}\nu_1\left(\frac{z}{4}\right)^\delta z}} {1-e^{-w\nu_1\left(\frac{z}{4}\right)^\delta}}
+ \frac{NK e^{-\frac{\nu_1 \left(\frac{z}{4}\right)^\delta }{72K^2} z}} {1-e^{-\frac{\nu_1}{36K^2}\left(\frac{z}{2}\right)^\delta}}.
\end{split}
\end{equation}
In the second phase, we investigate the probability that the optimal strategy profile is not visited frequently.
Let $v^{z\ast}=[\boldsymbol{a}^{k\ast}, C^N]$ be the optimal strategy profile in the $z$-th game phase and $F_z(v^{\ast})$ be the number of times the optimal strategy profile has been visited at the last $\lfloor \frac{z}{2}\rfloor +1$ game phases.
According to \cite{bistritz2018game} (see \emph{Lemma 16}), we have
\begin{equation}\label{A4}
F_z (v^{\ast}) \triangleq \sum_{i=z-\lfloor\frac{z}{2}\rfloor}^{z} \sum_{t\in \mathcal{G}_z} \mathbb{I} \left( v(t) = v^{i\ast} \right),  \ \forall k \in \mathcal{K}.
\end{equation}
Denote the stationary distribution of the optimal strategy profile by $\pi_{v^{\ast}} = \min\limits_{z-\lfloor\frac{z}{2}\rfloor\leq j \leq z} \pi_{v^{i\ast}}$.
If $0<\eta<\frac{1}{2}$, then $\pi_{v^{\ast}}>\frac{1}{2(1-\eta)}$ for a sufficiently large $z$, we have
\begin{equation}\label{A5}
\begin{split}
P_{g,z} &\triangleq \mathrm{Pr} \left(F_z (v^{\ast}) \leq \frac{1}{2} \sum_{i=z-\lfloor\frac{z}{2}\rfloor}^{z} \nu_2 i^\delta  \right) \\
&\leq \left( C_0 e^{- \frac{\nu_2 \eta^2}{144T_m(\frac{1}{8})} \left(\pi_{v^{\ast}} - \frac{1}{2(1-\eta)}\right)\left(\frac{z}{2}\right)^\delta} \right)^z,
\end{split}
\end{equation}
where $C_0$ is a constant \com{and} independent of $z, \pi_{v^{\ast}}$ and $\eta$.

Second, we bound the probability that event $\left(\overline{E^{k^\ast}_z}, E^{m^\ast}_z \right)$ holds.
The method is based on the regret analysis of the TS algorithm \cite{agrawal2013further}.
Here, event $\overline{E^{k^\ast}_z}$ means that the player found the optimal RIS at the end of the $z$-th game phase.
Let $P^n_{t,z}$ be the probability that player $n$ fails to find the best SF.
Since players can find the optimal SF in the third phase only when event $\overline{E^{k^\ast}_z}$ holds, we have
\begin{equation}\label{AA4}\small 
\begin{split}
\mathrm{Pr} \left(E^{m^\ast}_z| \overline{E^{k^\ast}_z} \right)
&\triangleq  \sum_{n=1}^{N} \mathrm{Pr} \left( \sum_{m=1,m\neq m^\ast}^{M} \sum_{j=1}^{z} W^j_{n,m} \geq \frac{1}{2} \sum_{i=1}^{z} \nu_3 2^i \right)\\
&\overset{(a)} \leq \sum_{n=1}^{N} 2^{- D_\mathrm{kl}\left( \left(1- \frac{c_m^{\ast}\theta^{\ast}_{n,m}}{\sum_{m=1}^{M} c_m \theta_{n,m}} \right) \| \frac{1}{2} \right)\sum_{i=1}^{z} \nu_3 2^i} \\
&\overset{(b)} \leq \sum_{n=1}^{N} 2^{-2\left(\frac{c_m^{\ast}\theta^{\ast}_{n,m}}{\sum_{m=1}^{M}c_m \theta_{n,m}} - \frac{1}{2}  \right)^2 \left( 2^{z+1} - 2 \right) \nu_3 }\\
&\overset{(c)} \leq N 2^{- \frac{(M-2)^2 \left( 2^{z} - 1 \right)\nu_3}{M^2} },
\end{split}
\end{equation}
where $D_\mathrm{kl}$ is the KL-divergence and $W^j_{n,m^{\ast}}$ is the number of times that the $m$-th suboptimal SF has been selected by player $n$ up to the $j$-th epoch.
Inequality (a) holds by using the large deviation theory \cite{cover2012elements}. 
Inequality (b) follows from Pinsker's inequality, i.e., $D_\mathrm{kl}(p \| q)\geq 2(p-q)^2$.
\com{Inequality (c) holds due to $Mc_m^{\ast}\theta^{\ast}_{n,m} \geq \sum_{m=1}^{M}c_m \theta_{n,m}$,
considering the worst case that each SF has the same probability of being selected.}
Therefore, by using $\mathrm{Pr} \left(\overline{E^{k^\ast}_z}, E^{m^\ast}_z \right) = \mathrm{Pr} \left(E^{m^\ast}_z|\overline{E^{k^\ast}_z}\right) \mathrm{Pr}\left( \overline{E^{k^\ast}_z} \right)$, we have \eqref{eepdfh1} which is given in the top of next page.
\newcounter{MYtempeqncnt3}
\setcounter{MYtempeqncnt3}{\value{equation}}
\setcounter{equation}{39}
\begin{figure*}[!ht]
\normalsize
\begin{equation}\label{eepdfh1}\small
\begin{split}
\mathrm{Pr} \left(\overline{E^{k^\ast}_z}, E^{m^\ast}_z \right)
\leq N 2^{- \frac{(M-2)^2 \left( 2^{z} - 1 \right)\nu_3}{M^2} }
\left(1 - \frac{2NK e^{-\frac{w}{2}\nu_1\left(\frac{z}{4}\right)^\delta z}} {1-e^{-w\nu_1\left(\frac{z}{4}\right)^\delta}}
+ \frac{NK e^{-\frac{\nu_1 \left(\frac{z}{4}\right)^\delta }{72K^2} z}} {1-e^{-\frac{\nu_1}{36K^2}\left(\frac{z}{2}\right)^\delta}}
+ \left( C_0 e^{- \frac{\nu_2 \eta^2}{144T_m(\frac{1}{8})} \left(\pi_{v^{\ast}} - \frac{1}{2(1-\eta)}\right)\left(\frac{z}{2}\right)^\delta} \right)^z
\right).
\end{split}
\end{equation}
\setcounter{equation}{\value{equation}}
\hrulefill
\end{figure*}

Then, we continue to bound $\overline{Reg}_z$ based on \eqref{A3}, \eqref{A5} and \eqref{eepdfh1}.
For $z>z_0$, we have
\begin{equation}\label{A6}
\footnotesize
\begin{split}
&\overline{Reg}_z  \leq N \Gamma_{\max} \nu_2 z^{\delta}  +  \mathrm{Pr}(E_z) N \Gamma_{\max}\nu_1 z^{\delta} +  \mathrm{Pr}(E_z) N\Gamma_{\max}\nu_3 2^z  \\
&\leq N\Gamma_{\max}\nu_2 z^{\delta} + N\Gamma_{\max}\left(\nu_1 z^{\delta} +\nu_3 2^z  \right)
\left( \frac{2NK e^{-\frac{w}{2}\nu_1\left(\frac{z}{4}\right)^\delta z}} {1-e^{-w\nu_1\left(\frac{z}{4}\right)^\delta}} \right.\\
&\left.+ \frac{NK e^{-\frac{\nu_1 \left(\frac{z}{4}\right)^\delta }{72K^2} z}} {1-e^{-\frac{\nu_1}{36K^2}\left(\frac{z}{2}\right)^\delta}}
+\left( C_0 e^{- \frac{\nu_2 \eta^2}{144T_m(\frac{1}{8})} \left(\pi_{v^{\ast}} - \frac{1}{2(1-\eta)}\right)\left(\frac{z}{2}\right)^\delta} \right)^z\right)+ \\
& N^2\Gamma_{\max} \left(\nu_1 z^{\delta} +  \nu_3 2^z\right) 2^{- \frac{(M-2)^2 \left( 2^{z} - 1 \right)\nu_3}{M^2} }
\left(1 - \frac{2NK e^{-\frac{w}{2}\nu_1\left(\frac{z}{4}\right)^\delta z}} {1-e^{-w\nu_1\left(\frac{z}{4}\right)^\delta}} \right.\\
&\left.+ \frac{NK e^{-\frac{\nu_1 \left(\frac{z}{4}\right)^\delta }{72K^2} z}} {1-e^{-\frac{\nu_1}{36K^2}\left(\frac{z}{2}\right)^\delta}}
+ \left( C_0 e^{- \frac{\nu_2 \eta^2}{144T_m(\frac{1}{8})} \left(\pi_{v^{\ast}} - \frac{1}{2(1-\eta)}\right)\left(\frac{z}{2}\right)^\delta} \right)^z \right)\\
&\leq N\Gamma_{\max}\left(\frac{\nu_1}{2}+3NK+\nu_2 \right)z^{\delta} + N\Gamma_{\max}(6NK+1)\nu_3\\
&+N^2\Gamma_{\max} \left(\nu_1 z^\delta + \nu_3 2^{z} \right) 2^{- \nu_3 \left(2^{z}-1\right)}.
\end{split}
\end{equation}
where $\Gamma_{\max} = \max_{n,i} \mu_{n,i}$ is the maximum real expected reward among all players' arms.
The first inequality holds since we consider the worst case that each player contributes the maximum regret $\Gamma_{\max}$.
The second inequality follows by using \eqref{A3} and \eqref{A5}.
The last inequality establishes on the facts that, for $z>z_0$,
\begin{equation}\small
\begin{split}
&\max \left\{C_0 e^{- \frac{\nu_2 \eta^2}{144T_m(\frac{1}{8})} \left(\pi_{v^{\ast}} - \frac{1}{2(1-\eta)}\right)\left(\frac{z}{2}\right)^\delta}, e^{-\frac{w}{2}\nu_1\left(\frac{z}{4}\right)^\delta}, e^{-\frac{\nu_1 \left(\frac{z}{4}\right)^\delta }{72K^2}}   \right\} \\
&<\frac{1}{2}
\end{split}
\end{equation}
and
\begin{equation}
2^{- \frac{(M-2)^2 \left( 2^{z} - 1 \right)\nu_3}{M^2} } \leq 2^{- \nu_3 \left(2^{z}-1\right)}.
\end{equation}

Finally, let $Z$ be the total number of epochs.
The total pseudo-regret $\overline{Reg}^{(1)}$ in Pattern \RNum{1} can be bounded as
\begin{equation}\label{A7}\small
\begin{split}
&\overline{Reg}^{(2)} \overset{(a)}{\leq} \sum_{z=1}^{Z} \overline{Reg}_z \overset{(b)}{\leq} N\Gamma_{\max} \sum_{z=1}^{z_0} \left( (\nu_1+\nu_2) z^{\delta} + \nu_3 2^z \right) \\
&+ N\Gamma_{\max} \sum_{z=z_0+1}^{Z} \left( N\Gamma_{\max}\left(\frac{\nu_1}{2}+3NK+\nu_2 \right)z^{\delta}\right. \\
&\left.+ N\Gamma_{\max}(6NK+1)\nu_3 +N^2\Gamma_{\max} \left(\nu_1 z^\delta + \nu_3 2^{z} \right) 2^{- \nu_3 \left(2^{z}-1\right)} \right) \\
&\overset{(c)}{\leq} N\Gamma_{\max} \sum_{z=1}^{z_0} \nu_3 2^z + N\Gamma_{\max} \sum_{z=1}^{Z}  (\nu_1+\nu_2) z^{\delta} + ZN\Gamma_{\max}(6NK+1)\nu_3 \\
&\overset{(d)}{\leq} N\Gamma_{\max}(\nu_1+\nu_2) \log^{1+\delta}_2 \left(\frac{T}{\nu_3}+2\right)  + N\Gamma_{\max} \nu_3 2^{z_0+1} \\
& + N\Gamma_{\max}(6NK+1)\nu_3 \log_2 \left(\frac{T}{\nu_3}+2\right) = O(\log^{1+\delta}_2 T),
\end{split}
\end{equation}
\com{where (b) holds since \eqref{A6} for $z> z_0$ and the worst case of $\Gamma_{\max}$ pre-round regret for $z\leq z_0$.}
Inequality (d) follows from $\sum_{z=1}^{Z}z^{\delta}\leq Z^{1+\delta}$ and the fact that $T\geq \sum_{z=1}^{Z-1}\nu_3 2^z \geq \nu_3 (2^Z-2)$,
which gives $Z^{1+\delta} \leq \log^{1+\delta}_2\left({T}/{\nu_3}+2\right)$.

For the Pattern \RNum{2}, the total pseudo-regret $\overline{Reg}^{(2)}$ can be bounded according to the regret analysis of the MTS algorithm in \cite{gupta2018low}, i.e.,
\begin{equation}\label{AA0}
\overline{Reg}^{(2)} \leq P_a (1+\varpi) \sum_{n=1}^{N} \sum_{a_n\in \mathcal{M}} \frac{\log_2 T}{\mathrm{D}_\mathrm{KL}(a_n, a_n^{\ast})} \Delta_{n,a_n},
\end{equation}
where $\varpi \in (0,1)$ and $\mathrm{D}_\mathrm{KL}(\cdot)$ is the KL-divergence.
Term $P_a$ is the active probability of the \com{UE}.

To sum up, the total pseudo-regret $\overline{Reg}$ of Algorithm \ref{GoT} is given by
\begin{equation}\label{A8}
\begin{split}
\overline{Reg} = &\overline{Reg}^{(1)}+ \overline{Reg}^{(2)} \\
\leq &N\Gamma_{\max} (1-P_a) \left( 2(\nu_1 +\nu_2) \log^{1+\delta}_2 \left(\frac{T}{\nu_3}+2\right)\right.\\
&\left.+(6NK+1)\nu_3 \log_2 \left(\frac{T}{\nu_3}+2\right)\right) \\
&+  P_a (1+\varpi) \sum_{n=1}^{N} \sum_{a_n\in \mathcal{M}} \frac{\log_2 T}{\mathrm{D}_\mathrm{KL}(a_n, a_n^{\ast})} \Delta_{n,a_n},
\end{split}
\end{equation}

\section{RIS's Direction and Element's Location}\label{appendix2}
We first determine the direction of the RIS in $XY$-plane by computing the angle $\angle \varphi$ between $X$-axis and the RIS, as shown in Fig. \ref{RS}.
Given the coordinates of $\mathrm{B} = \left(x_\mathrm{B}, y_\mathrm{B}\right)$, $\mathrm{R}=\left(x_\mathrm{R}, y_\mathrm{R}\right)$, $\mathrm{U}=\left(x_\mathrm{U}, y_\mathrm{U}\right)$,
we have two vectors $\overrightarrow{\mathrm{RB}} = \left(x_\mathrm{B} - x_\mathrm{R}, y_\mathrm{B}-y_\mathrm{R}\right)$ and $\overrightarrow{\mathrm{RU}} = \left(x_\mathrm{U}-x_\mathrm{R}, y_\mathrm{U}-y_\mathrm{R}\right)$.
According to the plane analytical geometry theory, we can obtain the direction vector $\overrightarrow{\mathrm{RC}}$, i.e., the bisector of angle $\angle \text{BRU}$, as
\begin{enumerate}
  \item If $\cos\langle \overrightarrow{\mathrm{RB}}, \overrightarrow{\mathrm{RU}} \rangle \geq 0$, then
        \begin{equation*}
        \begin{split}
        \overrightarrow{\mathrm{RC}} &= \left(x_{\mathrm{RC}}, y_{\mathrm{RC}} \right)  = -\frac{\overrightarrow{\mathrm{RB}}}{|\overrightarrow{\mathrm{RB}}|} + \frac{\overrightarrow{\mathrm{RU}}}{|\overrightarrow{\mathrm{RU}}|} \\
        & = \left(\frac{x_\mathrm{R} - x_\mathrm{B}}{|\overrightarrow{\mathrm{RB}}|}+ \frac{x_\mathrm{U} - x_\mathrm{R}}{|\overrightarrow{\mathrm{RU}}|}, \frac{y_\mathrm{R} - y_\mathrm{B}}{|\overrightarrow{\mathrm{RB}}|}+ \frac{y_\mathrm{U} - y_\mathrm{R}}{|\overrightarrow{\mathrm{RU}}|}  \right);
        \end{split}
        \end{equation*}
  \item If $\cos\langle \overrightarrow{\mathrm{RB}}, \overrightarrow{\mathrm{RU}} \rangle < 0$, then
        \begin{equation*}
        \begin{split}
        \overrightarrow{\mathrm{RC}} &= \left(x_{\mathrm{RC}}, y_{\mathrm{RC}} \right)  = \frac{\overrightarrow{\mathrm{RB}}}{|\overrightarrow{\mathrm{RB}}|} + \frac{\overrightarrow{\mathrm{RU}}}{|\overrightarrow{\mathrm{RU}}|} \\
        & = \left(\frac{x_\mathrm{B} - x_\mathrm{R}}{|\overrightarrow{\mathrm{RB}}|}+ \frac{x_\mathrm{U} - x_\mathrm{R}}{|\overrightarrow{\mathrm{RU}}|}, \frac{y_\mathrm{B} - y_\mathrm{R}}{|\overrightarrow{\mathrm{RB}}|}+ \frac{y_\mathrm{U} - y_\mathrm{R}}{|\overrightarrow{\mathrm{RU}}|}  \right).
        \end{split}
        \end{equation*}
\end{enumerate}
Thus, the direction of the RIS in $XY$-plane (i.e., the normal vector $\overrightarrow{\mathrm{AR}}$ of line RC) is $\overrightarrow{\mathrm{AR}} = \left(-y_{\mathrm{RC}},  x_{\mathrm{RC}}\right)$.
It is easy to obtain the angle $\angle \varphi$ by
\begin{equation}
\angle \varphi = -\arctan \left(\frac{x_{\mathrm{RC}}}{y_{\mathrm{RC}}} \right).
\end{equation}

Next, based on the angle $\angle \varphi$, we can compute the location of each element in the RIS, i.e.,
\begin{equation}\label{Locations}
\renewcommand{\arraystretch}{1.2}
\begin{split}
\left\{
\begin{array}{ll}
x(l_1, l_2) = \left(l_1 - 51\right) d_v \cos \angle \varphi + x_\mathrm{R},  \\
y(l_1, l_2) = \left(l_1 - 51\right) d_v \sin \angle \varphi + y_\mathrm{R},  \\
z(l_1, l_2) = \left(l_2 - 51\right) d_h  + 10,
\end{array} \right.
\end{split}
\end{equation}
where $d_v=d_h=0.01$ are the offsets in RIS's horizontal and vertical planes, respectively.
Constant $51$ is the $51$-th row or column elements in the RIS and constant $10$ is the height of the RIS.
Symbol $(l_1, l_2)$ are the integers in $[0,101]$, standing for the index ceil of the RIS elements' matrix.

\balance
\bibliographystyle{IEEEtran}
\bibliography{Reference_MPMAB}

\begin{thebibliography}{10}
\providecommand{\url}[1]{#1}
\csname url@samestyle\endcsname
\providecommand{\newblock}{\relax}
\providecommand{\bibinfo}[2]{#2}
\providecommand{\BIBentrySTDinterwordspacing}{\spaceskip=0pt\relax}
\providecommand{\BIBentryALTinterwordstretchfactor}{4}
\providecommand{\BIBentryALTinterwordspacing}{\spaceskip=\fontdimen2\font plus
\BIBentryALTinterwordstretchfactor\fontdimen3\font minus \fontdimen4\font\relax}
\providecommand{\BIBforeignlanguage}[2]{{%
\expandafter\ifx\csname l@#1\endcsname\relax
\typeout{** WARNING: IEEEtran.bst: No hyphenation pattern has been}%
\typeout{** loaded for the language `#1'. Using the pattern for}%
\typeout{** the default language instead.}%
\else
\language=\csname l@#1\endcsname
\fi
#2}}
\providecommand{\BIBdecl}{\relax}
\BIBdecl

\bibitem{wu2019towards}
Q.~Wu and R.~Zhang, ``Towards smart and reconfigurable environment: {Intelligent} reflecting surface aided wireless network,'' \emph{IEEE Commun. Mag.}, vol.~58, no.~1, pp. 106--112, Feb. 2019.

\bibitem{di2020smart}
M.~Di~Renzo, A.~Zappone, M.~Debbah, M.-S. Alouini, C.~Yuen, J.~De~Rosny, and S.~Tretyakov, ``Smart radio environments empowered by reconfigurable intelligent surfaces: How it works, state of research, and the road ahead,'' \emph{IEEE J. Select. Areas in Commun.}, vol.~38, no.~11, pp. 2450--2525, Nov. 2020.

\bibitem{MohamedRISwirelesscommun}
M.~A. ElMossallamy, H.~Zhang, L.~Song, K.~G. Seddik, Z.~Han, and G.~Y. Li, ``Reconfigurable intelligent surfaces for wireless communications: Principles, challenges, and opportunities,'' \emph{IEEE Trans. on Cogn. Commun. and Net.}, vol.~6, no.~3, pp. 990--1002, Mar. 2020.

\bibitem{zhang2020reconfigurable}
H.~Zhang, B.~Di, L.~Song, and Z.~Han, ``Reconfigurable intelligent surfaces assisted communications with limited phase shifts: {How} many phase shifts are enough?'' \emph{IEEE Trans. Veh. Technol.}, vol.~64, no.~4, pp. 4498--4502, Apr. 2020.

\bibitem{di2019hybrid}
B.~Di, H.~Zhang, L.~Song, Y.~Li, Z.~Han, and H.~V. Poor, ``Hybrid beamforming for reconfigurable intelligent surface based multi-user communications: Achievable rates with limited discrete phase shifts,'' vol.~38, no.~8, pp. 1809--1822, Aug. 2020.

\bibitem{li2020reconfigurable}
S.~Li, B.~Duo, X.~Yuan, Y.-C. Liang, and M.~Di~Renzo, ``Reconfigurable intelligent surface assisted {UAV} communication: Joint trajectory design and passive beamforming,'' vol.~9, no.~5, pp. 716--720, May 2020.

\bibitem{HonglinagBook}
H.~Zhang, B.~Di, L.~Song, and Z.~Han, \emph{Reconfigurable Intelligent Surface-Empowered 6G}.\hskip 1em plus 0.5em minus 0.4em\relax Springer, 2021.

\bibitem{liberg2017cellular}
O.~Liberg, M.~Sundberg, E.~Wang, J.~Bergman, and J.~Sachs, \emph{Cellular {Internet} of things: technologies, standards, and performance}.\hskip 1em plus 0.5em minus 0.4em\relax Academic Press, 2017.

\bibitem{qi2018wireless}
Q.~Qi and X.~Chen, ``Wireless powered massive access for cellular {Internet of Things} with imperfect {SIC} and nonlinear {EH},'' \emph{IEEE Internet of Things J.}, vol.~6, no.~2, pp. 3110--3120, Feb. 2018.

\bibitem{dama2016feasible}
S.~Dama, V.~Sathya, K.~Kuchi, and T.~V. Pasca, ``A feasible cellular {Internet of Things}: Enabling edge computing and the iot in dense futuristic cellular networks,'' \emph{IEEE Consumer Electron. Mag.}, vol.~6, no.~1, pp. 66--72, Jan. 2016.

\bibitem{elsaadany2017cellular}
M.~Elsaadany, A.~Ali, and W.~Hamouda, ``Cellular {LTE-A} technologies for the future {Internet-of-Things}: {Physical} layer features and challenges,'' \emph{IEEE Commun. surveys Tuts.}, vol.~19, no.~4, pp. 2544--2572, Apr. 2017.

\bibitem{waret2018lora}
A.~Waret, M.~Kaneko, A.~Guitton, and N.~El~Rachkidy, ``{LoRa} throughput analysis with imperfect spreading factor orthogonality,'' \emph{IEEE Wireless Commun. Lett.}, vol.~8, no.~2, pp. 408--411, Oct. 2018.

\bibitem{lyu2019achieving}
J.~Lyu, D.~Yu, and L.~Fu, ``Achieving max-min throughput in {LoRa} networks,'' in \emph{Int. Conf. on Comput., Net. and Commun.}, Big Island, HI, Feb. 2020.

\bibitem{boyd2004convex}
S.~Boyd, S.~P. Boyd, and L.~Vandenberghe, \emph{Convex optimization}.\hskip 1em plus 0.5em minus 0.4em\relax Cambridge University Press, 2004.

\bibitem{papadimitriou1998combinatorial}
C.~H. Papadimitriou and K.~Steiglitz, \emph{Combinatorial optimization: algorithms and complexity}.\hskip 1em plus 0.5em minus 0.4em\relax Courier Corporation, 1998.

\bibitem{nasir2019multi}
Y.~S. Nasir and D.~Guo, ``Multi-agent deep reinforcement learning for dynamic power allocation in wireless networks,'' \emph{IEEE J. Select. Areas in Commun.}, vol.~37, no.~10, pp. 2239--2250, Oct. 2019.

\bibitem{gu2020deep}
B.~Gu, X.~Zhang, Z.~Lin, and M.~Alazab, ``Deep multiagent reinforcement-learning-based resource allocation for {Internet} of controllable things,'' \emph{IEEE Int. of Things J.}, vol.~8, no.~5, pp. 3066--3074, May 2020.

\bibitem{sutton2018reinforcement}
R.~S. Sutton and A.~G. Barto, \emph{Reinforcement learning: {An} introduction}.\hskip 1em plus 0.5em minus 0.4em\relax MIT Press, 2018.

\bibitem{bubeck2012regret}
S.~Bubeck and N.~Cesa-Bianchi, ``Regret analysis of stochastic and nonstochastic multi-armed bandit problems,'' \emph{Foundations and Trends{\textregistered} in Mach. Learning}, vol.~5, no.~1, pp. 1--122, Dec. 2012.

\bibitem{ta2019lora01}
D.-T. Ta, K.~Khawam, S.~Lahoud, C.~Adjih, and S.~Martin, ``{LoRa-MAB}: Toward an intelligent resource allocation approach for {LoRaWAN},'' in \emph{Proc. IEEE Glob. Telecom. Conf.}, Waikoloa, HI, Dec. 2019.

\bibitem{tibrewal2019distributed}
H.~Tibrewal, S.~Patchala, M.~K. Hanawal, and S.~J. Darak, ``Distributed learning and optimal assignment in multiplayer heterogeneous networks,'' in \emph{Proc. IEEE INFOCOM}, Pairs, France, Jun. 2019.

\bibitem{zafaruddin2019distributed}
S.~Zafaruddin, I.~Bistritz, A.~Leshem, and D.~Niyato, ``Distributed learning for channel allocation over a shared spectrum,'' \emph{IEEE J. Select. Areas in Commun.}, vol.~37, no.~10, pp. 2337--2349, Aug. 2019.

\bibitem{bistritz2018game}
I.~Bistritz and A.~Leshem, ``Distributed multi-player bandits-a game of thrones approach,'' in \emph{Advances in Neural Inform. Process. Syst.}, Montréal, Canada, Dec. 2018.

\bibitem{Exp3}
P.~{Auer}, N.~{Cesa-Bianchi}, Y.~{Freund}, and R.~E. {Schapire}, ``The nonstochastic multi-armed bandit problem,'' \emph{SIAM Journal on Computing}, vol.~32, no.~1, pp. 48--77, Jan. 2002.

\bibitem{jonker1986improving}
R.~Jonker and T.~Volgenant, ``Improving the {Hungarian} assignment algorithm,'' \emph{Operations Res. Letters}, vol.~5, no.~4, pp. 171--175, Apr. 1986.

\bibitem{auer2002finite}
P.~Auer, N.~Cesa-Bianchi, and P.~Fischer, ``Finite-time analysis of the multiarmed bandit problem,'' \emph{Mach. Learning}, vol.~47, no. 2-3, pp. 235--256, Feb. 2002.

\bibitem{arras2019bound}
B.~Arras, E.~Azmoodeh, G.~Poly, and Y.~Swan, ``A bound on the {Wasserstein}-2 distance between linear combinations of independent random variables,'' \emph{Stochast. Processes and Their Appl.}, vol. 129, no.~7, pp. 2341--2375, Jul. 2019.

\bibitem{gupta2018low}
H.~Gupta, A.~Eryilmaz, and R.~Srikant, ``Low-complexity, low-regret link rate selection in rapidly-varying wireless channels,'' in \emph{Proc. IEEE INFOCOM}, Honolulu, HI, Apr. 2018.

\bibitem{agrawal2013further}
S.~Agrawal and N.~Goyal, ``Further optimal regret bounds for {Thompson} sampling,'' in \emph{Artificial Intell. and Statist.}, Scottsdale, AZ, Apr. 2013.

\bibitem{you2021reconfigurable}
L.~You, J.~Xiong, Y.~Huang, D.~W.~K. Ng, C.~Pan, W.~Wang, and X.~Gao, ``Reconfigurable intelligent surfaces-assisted multiuser {MIMO} uplink transmission with partial {CSI},'' \emph{IEEE Trans. Wireless Commun.}, vol.~20, no.~9, pp. 5613--5627, Sep. 2017.

\bibitem{abeywickrama2020intelligent}
S.~Abeywickrama, R.~Zhang, Q.~Wu, and C.~Yuen, ``Intelligent reflecting surface: Practical phase shift model and beamforming optimization,'' \emph{IEEE Trans. on Commun.}, vol.~68, no.~9, pp. 5849--5863, Sep. 2020.

\bibitem{3GPP}
{3GPP TR 38.901}, ``Study on channel model for frequencies from 0.5 to 100 {GHz} (release 14),'' \emph{3GPP}, Jan. 2018.

\bibitem{tong2018cooperative}
J.~Tong, M.~Jin, Q.~Guo, and Y.~Li, ``Cooperative spectrum sensing: A blind and soft fusion detector,'' \emph{IEEE Trans. Wireless Commun.}, vol.~17, no.~4, pp. 2726--2737, Apr. 2018.

\bibitem{tong2017energy}
J.~Tong, M.~Jin, Q.~Guo, and L.~Qu, ``Energy detection under interference power uncertainty,'' \emph{IEEE Commun. Lett.}, vol.~21, no.~8, pp. 1887--1890, Aug. 2017.

\bibitem{VL1}
J.~{Choi}, C.~{Joo}, J.~{Zhang}, and N.~B. {Shroff}, ``Distributed link scheduling under {SINR} model in multihop wireless networks,'' \emph{IEEE/ACM Trans. Networking.}, vol.~22, no.~4, pp. 1204--1217, Aug. 2014.

\bibitem{Gupta2018Low-complexity}
H.~{Gupta}, A.~{Eryilmaz}, and R.~{Srikant}, ``Low-complexity, low-regret link rate selection in rapidly-varying wireless channels,'' in \emph{Proc. IEEE INFOCOM}, Honolulu, HI, Apr. 2018.

\bibitem{menon2013convergence}
A.~Menon and J.~S. Baras, ``Convergence guarantees for a decentralized algorithm achieving {Pareto} optimality,'' in \emph{American Control Conf.}, Washington, DC, Jun. 2013.

\bibitem{marden2014achieving}
J.~R. Marden, H.~P. Young, and L.~Y. Pao, ``Achieving {Pareto} optimality through distributed learning,'' \emph{SIAM J. on Control and Optimization}, vol.~52, no.~5, pp. 2753--2770, May 2014.

\bibitem{rubner2000earth}
Y.~Rubner, C.~Tomasi, and L.~J. Guibas, ``The earth mover's distance as a metric for image retrieval,'' \emph{Int. J. of Comput. Vision}, vol.~40, no.~2, pp. 99--121, Nov. 2000.

\bibitem{likas2003global}
A.~Likas, N.~Vlassis, and J.~J. Verbeek, ``The global $k$-means clustering algorithm,'' \emph{Pattern Recognition}, vol.~36, no.~2, pp. 451--461, Feb. 2003.

\bibitem{cover2012elements}
T.~M. Cover and J.~A. Thomas, \emph{Elements of information theory}.\hskip 1em plus 0.5em minus 0.4em\relax John Wiley \& Sons, 2012.

\end{thebibliography}
\end{document}